\documentclass[11pt,a4paper]{article}
\usepackage{amsmath, amssymb, amsthm}
\usepackage{fullpage}
\usepackage{sgame}
\usepackage{color}

\usepackage{ifpdf}
\ifpdf
  \usepackage[pdftex]{hyperref}
\else
  \usepackage[hypertex]{hyperref}
\fi

\newtheorem{theorem}{Theorem}[section]
\newtheorem{definition}{Definition}[section]
\newtheorem{lemma}[theorem]{Lemma}
\newtheorem{cor}[theorem]{Corollary}
\newtheorem{obs}[theorem]{Observation}
\newtheorem{prop}[theorem]{Proposition}

\newcommand{\Prob}[2]{\mathbf{P}_{#1} \left( #2 \right)}
\newcommand{\mobs}[2]{{\langle #1,#2\rangle}}

\newcommand{\x}{{\mathbf{x}}}
\newcommand{\y}{{\mathbf y}}
\newcommand{\z}{{\mathbf z}}
\newcommand{\w}{{\mathbf w}}

\newcommand{\tx}{{\tilde{x}}}
\newcommand{\txx}{{\tilde{\x}}}
\newcommand{\ty}{{\tilde{y}}}
\newcommand{\tyy}{{\tilde{\y}}}

\newcommand{\tm}{{t_{\rm mix}}}

\renewcommand{\leq}{\leqslant}
\renewcommand{\geq}{\geqslant}
\renewcommand{\epsilon}{\varepsilon}

\newcommand{\tv}[1]{\left\|#1\right\|_{\rm TV}}
\newcommand{\OO}{{\mathcal O}}
\newcommand{\calM}{{\mathcal M}}
\newcommand{\G}{\mathcal{G}}

\newcommand{\Diff}{\mathsf{Diff}}
\newcommand{\MonoC}{\mathsf{MonoC}}

\renewcommand{\deg}{\mathsf{deg}}

\title{Logit Dynamics with Concurrent Updates for Local Interaction Games\thanks{Vincenzo Auletta and Giuseppe Persiano are supported by Italian MIUR under the PRIN 2010-2011 project \emph{ARS TechnoMedia -- Algorithmics for Social Technological Networks}. Diodato Ferraioli and Francesco Pasquale are supported by EU FET project MULTIPLEX 317532.}}

\author{
Vincenzo Auletta\thanks{Universit\`a di Salerno, Italy. Email: \texttt{auletta@dia.unisa.it}.}
\and
Diodato Ferraioli\thanks{``Sapienza'' Universit\`a di Roma, Italy. Email: \texttt{ferraioli@dis.uniroma1.it}.}
\and
Francesco Pasquale\thanks{``Sapienza'' Universit\`a di Roma, Italy. Email: \texttt{pasquale@di.uniroma1.it}.}
\and
Paolo Penna\thanks{Autonomous Researcher, Italy. Email: \texttt{paolo.penna@gmail.com}.}
\and
Giuseppe Persiano\thanks{Universit\`a di Salerno, Italy. Email: \texttt{giuper@dia.unisa.it}.}
}
\date{}

\begin{document}
\setcounter{page}{0}
\maketitle
\thispagestyle{empty}

\begin{abstract}
{\em Logit choice} dynamics are a family of randomized best response dynamics based on the \emph{logit choice function} \cite{McFadden74}
that is used for modeling players with limited rationality and knowledge.
In this paper we study the \emph{all-logit dynamics}, where at {\em each} time step {\em all} players {\em concurrently} update their strategies according to the logit choice function. In the well studied {\em one-logit dynamics} \cite{blumeGEB93}
instead at each step {\em only one} randomly chosen player is allowed to update.

We study properties of the all-logit dynamics in the context of \emph{local interaction games}, a class of games that has been used to model complex social phenomena \cite{blumeGEB93,youngTR00,msFOCS09}
and physical systems~\cite{llpPTRF10}.
In a local interaction game, players are the vertices of a {\em social graph} whose edges are two-player potential games. Each player picks one strategy to be played for all the games she is involved in and the payoff of the player is the sum of the payoffs from each of the games. We prove that local interaction games characterize the class of games for which the all-logit dynamics is reversible.

We then compare the stationary behavior of one-logit and all-logit dynamics.
Specifically, we look at the expected value of a notable class of observables, that we call
{\em decomposable} observables.
We prove that the difference between the expected values of the observables at stationarity
for the two dynamics depends only on the \emph{rationality} level $\beta$ and
on the distance of the social graph from  a bipartite graph.
In particular, if the social graph is bipartite then decomposable observables have the same expected
value.
Finally, we show that the mixing time of the all-logit dynamics has the same twofold behavior that has been highlighted in the case of the one-logit:
for some games it exponentially depends on the rationality level $\beta$, whereas for other games it can be upper bounded by a function independent from $\beta$.
\end{abstract}

\newpage
\pagenumbering{arabic}

\section{Introduction}\label{sec::intro}
In the last decade,
we have observed an increasing interest in understanding phenomena
occurring in complex systems consisting of a large number of simple
networked components that operate autonomously guided by their own objectives and influenced by the behavior of the neighbors.
Even though (online) social networks are a primary example of
such systems, other remarkable typical instances can be found in
Economics (e.g., markets),
Physics (e.g., Ising model and spin systems) and Biology (e.g., evolution of life).
A common feature of these systems is that the behavior of each component
depends only on the interactions with a limited
number of other components (its neighbors) and
these interactions are usually very simple.

Game Theory is the main tool used to model the behavior of agents
that are guided by their own objective
in contexts where their gains depend also on the choices made by
neighboring agents.
Game theoretic approaches have been often proposed for modeling phenomena
in a complex social network,
such as the formation of the social network itself
\cite{jwJET1996,bgECO2000,adtwSTOC2003,flmpsPODC2003,cpPODC2005,bckmrrsWAW2010,bcdlICS2011},
the formation of opinions~\cite{koSTOC2011,bkoFOCS2011,fgvSAGT2012}
and the spread of innovation~\cite{youngPUP98,youngTR00,msFOCS09}.
Many of these models are based on \emph{local interaction games}~\cite{morris},
where agents are represented as vertices on a {\em social graph} and
the relationship between two agents is represented by a simple
two-player game played on the edge joining the corresponding vertices.

We are interested in the \emph{dynamics} that govern such phenomena and
several dynamics have been studied in the literature like, for example,
the best response dynamics~\cite{ftMIT91},
the logit dynamics~\cite{blumeGEB93},
fictitious play~\cite{flMIT98}
or no-regret dynamics \cite{hmjet}.
Any such dynamics can be seen as made of two components:
\begin{itemize}
\item \emph{Selection rule:}
by which the set of players that update their
state (strategy) is determined;
\item \emph{Update rule:}
by which the selected players update their strategy.
\end{itemize}
For example,
the classical best response dynamics compose the {\em best response} update
rule with a selection rule that selects one player at the time.
In the best response update rule,
the selected player picks the strategy that, given
the current strategies of the other players, guarantees the highest utility.
The Cournot dynamics~\cite{cournot} instead combine the
best response update rule with the selection rule that selects all players.
Other dynamics in which all players concurrently update their strategy
are fictitious play~\cite{flMIT98} and the no-regret dynamics~\cite{hmjet}.

In this paper,
we study a specific class of randomized update rules called
the {\em logit choice function} \cite{McFadden74,blumeGEB93,wolpert} which
is a type of noisy best response that models in a clean and
tractable way the limited knowledge (or bounded rationality) of
the players in terms of a parameter $\beta$ called {\em inverse noise}.
In similar models studied in Physics,
$\beta$ is the inverse of the temperature.
Intuitively, a low value of $\beta$ (that is, high temperature)
models a noisy scenario in which players choose their strategies
``nearly at random'';
a high value of $\beta$ (that is, low temperature) models
a scenario with little noise in which players
pick the strategies yielding higher payoffs with higher probability.

The logit choice function can be coupled with different selection rules so to give
different dynamics.
For example, in the {\em logit} dynamics \cite{blumeGEB93}
at every time step a single player is selected uniformly at random and the
selected player updates her strategy according to the logit choice function.
The remaining players are not allowed to revise their strategies in this time step.
One of the appealing features of the logit dynamics is that it naturally describes an ergodic Markov chain.
This means that the underlying Markov chain admits a
\emph{unique} \emph{stationary distribution} which we take as solution concept.
This distribution describes the long-run behavior of the system (which states appear  more frequently
over a long run).
The interplay between the noise and the underlying game naturally determines the system behavior:
(i) As the noise becomes ``very large'' the equilibrium point is ``approximately'' the uniform distribution; (ii) As the noise vanishes the stationary distribution concentrates on so called stochastically stable states \cite{sandholmMIT10} which, for certain classes of games, correspond to pure Nash equilibria \cite{blumeGEB93, afnGEB10}.

While the logit choice function is a very natural behavioral model for approximately rational agents,
the specific selection rule that selects one single player per time step
avoids any form of concurrency.
Therefore a natural question arises
\begin{quote}
 What happens if \emph{concurrent} updates are allowed?
\end{quote}
For example, it is easy to construct games for which the best response converges to a Nash equilibrium
when only one player is selected at each step and does not converge to any state when more players are
chosen to concurrently update their strategies.

In this paper we study how the logit choice function behave in an extreme case of concurrency. Specifically, we couple this update rule
with a selection rule by which \emph{all} players update their strategies at every time step.
We call such dynamics \emph{all-logit}, as opposed to the classical (\emph{one-})logit dynamics
in which only one player at a time is allowed to move.
Roughly speaking,
the all-logit are to the one-logit what the Cournot dynamics are to the best
response dynamics.

\paragraph{Our contributions.}
We study the all-logit dynamics for local interaction games \cite{morris,ellisonECO93,msFOCS09}.
Here players are vertices of a graph, called the {\em social graph},
and each edge is a two-player (exact) potential game.
We remark that games played on different edges by a player may be different but,
nonetheless, they have the same strategy set for the player.
Each player picks one strategy that is used for all of
her edges and the payoff is a (weighted) sum of the payoffs obtained from each game.
This class of games includes coordination games on a network \cite{ellisonECO93}
that have been used to model the spread of innovation and of new technology in social networks \cite{youngPUP98,youngTR00}, and
the Ising model \cite{Mart1999}, a model for magnetism.
In particular, we study the all-logit dynamics on local interaction games
for every possible value of the inverse noise $\beta$ and we are interested on properties of the
original one-logit dynamics that are preserved by the all-logit.

As a warm-up, we discuss two classical two-player games
(these are trivial local interaction games played on a graph with two vertices and one edge):
the coordination game and
the prisoner's dilemma. Even though for both games
the stationary distribution of the one-logit and of the all-logit
are quite different, we identify three similarities.
First, for both games, both Markov chains are reversible.
Moreover, for both games,
the expected number of players playing a certain strategy at the stationarity of
the all-logit is exactly the same as if the expectation was taken
on the stationary distribution of the one-logit.
Finally, for these games the mixing time is asymptotically the same regardless of the selection rule.
In this paper we will show that none of these findings is accidental.

We first study the \emph{reversibility} of the all-logit dynamics,
an important property of stochastic processes that is useful
also to obtain explicit formulas for the stationary distribution.
We \emph{characterize}
the class of games for which the all-logit dynamics
(that is, the Markov chain resulting from the all-logit dynamics) are reversible and
it turns out that this class coincides with the class of local interaction games.
This implies that the all-logit dynamics of all two-player potential games are reversible;
whereas not all potential games have reversible all-logit dynamics.
This is to be compared with the well-known result saying that one-logit dynamics
of every potential game are reversible with respect to the Gibbs measure~\cite{blumeGEB93}.
One of the tools we develop for our characterization yields a closed formula for
the stationary distribution of reversible all-logit  dynamics.

Then, we focus on the \emph{observables} of local interaction games.
An observable is a function of the strategy profile
(that is the sequence of strategies adopted by the players) and
we are interested in its expected values at stationarity for both the one-logit and the all-logit.
A prominent example of observable is the difference $\Diff$ between the
number of players adopting two given strategies in a game.
In a local interaction game modeling the spread of innovation on a social network
this observable counts the difference between the
number of adopters of the new and old technology whereas in the Ising model
it is the magnetic field of a magnet.

We show that there exists a class of observables whose expectation at stationarity of
the all-logit is the same as the expectation at stationarity of the one-logit as long as
the social network underlying the local interaction game is bipartite
(and thus trivially for all two-player games).
This class of observables includes the $\Diff$ observable. It is interesting to note
that the Ising game has been mainly studied for bipartite graphs (e.g., the two-dimensional
and the three-dimensional lattice).
This implies that, for the Ising model,
the all-logit dynamics are compatible with the observations and it is arguably
more natural than the one-logit (that postulate that at any given time step only one particle
updates its status and that the update strategy is instantaneously propagated).
We extend this result by showing that for general graphs,
the extent at which the expectations of these observables differ
can be upper and lower bounded  by a function of $\beta$ and of the distance
of the social graph from a bipartite graph.

Finally, we give the first bounds on the mixing time of the all-logit.
We start by giving a {\em general} upper bound on the mixing time of the all-logit
in terms of the {\em cumulative utility} of the game.
We then look at two specific classes of games: graphical coordination games and games with a dominant profile.
For \emph{graphical coordination games}, we prove an upper bound to the mixing time that exponentially depends on $\beta$.
Note that it is known~\cite{afpppSPAA11} that the one-logit also take a time exponential in $\beta$ for converging to the stationary distribution.
For \emph{games with a dominant profile}, we instead prove that the mixing time can be bounded by a function independent from $\beta$.
Thus, also for these games the mixing time of the all-logit has the same behavior of the one-logit mixing time.

\paragraph{Related works on logit dynamics.}
The all-logit dynamics for strategic games have been studied by Alos-Ferrer and Netzer~\cite{afnGEB10}.
Specifically, in \cite{afnGEB10} the authors study the logit-choice function combined with {general} selection rules
(including the selection rule of the all-logit)
and investigate conditions
for which a state is \emph{stochastically stable}.
A stochastically stable state is a state that has non-zero probability
as $\beta$ goes to infinity \cite{sandholmMIT10}.
We focus instead on a specific selection rule that is used by several
remarkable dynamics considered
in Game Theory (Cournot, fictitious play, and no-regret)
and consider the whole range of values of $\beta$.

The one-logit dynamics have been actively studied starting from the
work of Blume~\cite{blumeGEB93}
that showed that for $2\times 2$ coordination games,
the risk dominant equilibria (see \cite{hsMIT88})
are stochastically stable.
Much work has been devoted to the study of the one-logit for
local interaction games with the aim of modeling and understanding
the spread of innovation in a social network~\cite{ellisonECO93,youngTR00}.
A general upper bound on the mixing time of the one-logit dynamics for
this class of games is given by Berger et al.~\cite{bkmpPTRF05}.
Montanari and Saberi \cite{msFOCS09} instead studied the hitting time
of the highest potential configuration and relate this quantity to a
connectivity property of the underlying network.
Asadpour and Saberi \cite{asWINE09} considered the same problem for congestion games.
The mixing time and the metastability of the one-logit
dynamics for strategic games have been
studied in~\cite{afpppSPAA11,afppSODA12}.

\section{Definitions}\label{sec::logit}
In this section we formally define the local interaction games and the Markov chain induced by the all-logit dynamics.

\paragraph{Strategic games.}
Let $\mathcal{G}=\left([n],S_1,\dots,S_n,u_1,\dots,u_n\right)$ be a finite
{normal-form strategic game}. The set $[n]=\{1,\ldots,n\}$ is the player set,
$S_i$ is the set of \emph{strategies} for player $i\in[n]$,
$S=S_1\times S_2\times\cdots\times S_n$ is the set of
\emph{strategy profiles} and
$u_i\colon S\rightarrow\mathbb{R}$ is
the \emph{utility} function of player $i \in [n]$.

We adopt the standard game-theoretic notation and denote by $S_{-i}$ the set $S_{-i}=S_1\times\ldots\times S_{i-1}\times S_{i+1}\times\ldots S_n$ and, for $\x=(x_1, \dots, x_{i-1}, x_{i+1}, \dots, x_n)\in S_{-i}$ and $y \in S_i$, we denote by $(\x,y)$ the strategy profile $(x_1, \dots, x_{i-1}, y, x_{i+1}, \dots, x_n) \in S$.
Also, for a subset $L\subseteq[n]$ and strategy profile 
$\x$, we denote by $\x_L$ the components of $\x$ 
corresponding to players in $L$.

\paragraph{Potential games.}
We say that function $\Phi \colon S\rightarrow\mathbb{R}$ is an \emph{exact potential} (or simply a \emph{potential}) for game $\G$ if for every $i\in[n]$ and every $\x \in S_{-i}$
$$u_i(\x,y)-u_i(\x,z)=\Phi(\x,z)-\Phi(\x,y)$$
for all $y,z\in S_i$.
A game $\G$ that admits a potential is called a \emph{potential game}~\cite{MS96}.

The following is an important characterization of potential games in terms of the utilities. 
A \emph{circuit} $\omega=\langle s_0,\ldots,s_\ell\rangle $ of length $\ell$ is a sequence 
of strategy profiles such that $s_0=s_\ell$, $s_h\ne s_k$ for $1\leq h\ne k\leq\ell$ and, for $k=1,\ldots,\ell$, 
there exists player $i_k$ such that $s_{k-1}$ and $s_{k}$ differ only for player $i_k$.
For such a circuit $\omega$ we define the \emph{utility improvement} $I(\omega)$ as
$$
I(\omega)=\sum_{k=1}^\ell \left[ u_{i_k}(s_k)-u_{i_k}(s_{k-1})\right].
$$
The following theorem then holds.
\begin{theorem}[{\cite[Thm~2.8]{MS96}}]
\label{MS96}
A game $\G$ is a potential game if and only if $I(\omega)=0$ for all circuits of length $4$.
\end{theorem}

\paragraph{Local interaction games.}
In a \emph{local interaction game} $\G$, each player $i$, with strategy set $S_i$, is represented by a vertex of a graph $G=(V,E)$ (called \emph{social graph}).
For every edge~$e=(i,j) \in E$ there is a two-players game $\G_e$ with potential function $\Phi_e$ in which the set of strategies of endpoints are exactly $S_i$ and $S_j$.
We denote with $u_i^e$ the utility function of player $i$ in the game $\G_e$.
Given a strategy profile $\x$, the utility function of player $i$ in the local interaction game $\G$ sets
$$
 u_i(\x) = \sum_{e=(i,j)} u_i^e(x_i,x_j).
$$
It is easy to check that the function $\Phi = \sum_e \Phi_e$ is a potential function for the local interaction game $\G$.
Note that we assume that the graph $G$ is unweighted. However, it is immediate to see that weights do not give any modeling power.

\paragraph{Logit choice function.}
We study the interaction of $n$ players of a strategic game
$\G$ that update their strategy according to the
\emph{logit choice function}~\cite{McFadden74,blumeGEB93,wolpert}
described as follows:
from profile $\x \in S$ player $i \in [n]$ updates her strategy
to $y \in S_i$ with probability
\begin{equation}\label{eq:logitrule}
\sigma_{i}(y \mid \x) =\frac{e^{\beta u_i(\x_{-i}, y)}}{\sum_{z \in S_i} e^{\beta u_i(\x_{-i},z)}}.
\end{equation}
In other words, the logit choice function leans towards strategies
promising higher utility. The parameter $\beta\geq 0$ is a measure of
how much the utility influences the choice of the player.

\paragraph{All-logit.}
In this paper we consider the \emph{all-logit} dynamics,
by which {\em all} players {\em concurrently} update their strategy
using the logit choice function.
Most of the previous works have focused on
dynamics where at each step {\em one} player is chosen
uniformly at random and
she updates her strategy by following the logit choice function.
We call those dynamics \emph{one-logit},
to distinguish it from the \emph{all-logit}.

The all-logit dynamics induce a Markov chain over the set of strategy profiles whose transition probability $P(\x,\y)$
from profile $\x = (x_1, \dots, x_n)$ to profile $\y = (y_1, \dots, y_n)$ is
\begin{equation}\label{eq:alllogitchain}
P(\x,\y) = \prod_{i=1}^n \sigma_i(y_i \mid \x)
= \frac{e^{\beta\sum_{i=1}^n u_i(\x_{-i},y_i)}} {\prod_{i=1}^n\sum_{z\in S_i} e^{\beta u_i(\x_{-i},z)}}.
\end{equation}
Sometimes it is useful to write the transition probability from $\x$ to $\y$ in terms of the
\emph{cumulative utility} of $\x$ with respect to $\y$ defined as
$U(\x,\y)=\sum_i u_i(\x_{-i},y_i)$.
Indeed, by observing that
$$
 \prod_{i=1}^{n} \sum_{z \in S_i} e^{\beta u_i(\x_{-i},z)} = \sum_{\z \in S} \prod_{i=1}^{n} e^{\beta u_i(\x_{-i},z_i)},
$$
we can rewrite \eqref{eq:alllogitchain} as
\begin{equation}\label{eq:transition}
 P(\x,\y) = \frac{e^{\beta U(\x,\y)}}{T(\x)},
\end{equation}
where $T(\x)=\sum_{\z \in S} e^{\beta U(\x,\z)}$.
For a potential game $\mathcal{G}$ with potential $\Phi$,
we define for each pair of profiles $(\x, \y)$ the quantity
\begin{equation}
 \label{eq:kappa}
 K(\x,\y) = \sum_i\Phi(\x_{-i},y_i) - (n-2) \Phi(\x) = 2 \Phi(\x) + \sum_i \left(\Phi(\x_{-i},y_i) - \Phi(\x)\right).
\end{equation}
Simple algebraic manipulations show that, for a potential game,
we can rewrite the transition probabilities in~\eqref{eq:transition} as
$$
P(\x,\y) = \frac{e^{-\beta K(\x,\y)}}{\gamma_A(\x)},
$$
where $\gamma_A(\x)=\sum_{\z \in S} e^{-\beta K(\x,\z)}$.

It is easy to see that a Markov chain with transition matrix \eqref{eq:alllogitchain} is ergodic. Indeed, for example, ergodicity follows from the fact that all entries of the transition matrix are strictly positive.

\paragraph{Reversibility, Observables, Mixing time.}
In this work we focus on three features of the all-logit dynamics, that we formally define here.

Let $\calM$ be a Markov chain with transition matrix $P$ and state set $S$.
$\cal M$ is \emph{reversible} with respect to a distribution $\pi$
if, for every pair of states $x,y\in S$, the following
{\em detailed balance condition} holds
\begin{equation}
\label{eq:rev_det_bal}
 \pi(x)P(x,y) = \pi(y) P(y,x).
\end{equation}
It is easy to see that if $\calM$ is reversible with respect to $\pi$ then $\pi$ is also stationary.

An {\em observable} $O$ is a function $O \colon S\rightarrow {\mathbb R}$, i.e. it is a function that assigns a value to each strategy profile of the game.

An ergodic Markov chain has a unique stationary distribution $\pi$ and for every starting profile $\x$ the distribution $P^t(\x,\cdot)$ of the chain at time $t$ converges to $\pi$ as $t$ goes to infinity. The \emph{mixing time} is a measure of how long it takes to get close to the stationary distribution from the \emph{worst-case} starting profile, and it is defined as
$$
\tm(\varepsilon) = \inf \left\{t \in \mathbb{N} \colon \tv{P^t(\x, \cdot) - \pi} \leqslant \varepsilon \mbox{ for all } \x \in S \right\},
$$
where $\tv{P^t(\x, \cdot) - \pi} = \frac{1}{2} \sum_{\y \in S} |P^t(\x, \y) - \pi(\y)|$ is the \emph{total variation distance}. We will usually use $\tm$ for $\tm(1/4)$. We refer the reader to~\cite{lpwAMS08} for a more detailed description of notational conventions about Markov chains and mixing times.

\section{Warm-up: two-player games}
\label{sec::examples}
\newsavebox\tpbox
\begin{lrbox}{\tpbox}
\begin{game}{2}{2}
 & $-$ & $+$ \\
$-$ & $a,a$ & $c,d$ \\
$+$ & $d,c$ & $b,b$
\end{game}
\end{lrbox}

In this section we compare the behavior of the one- and the all-logit dynamics for two simple two-player potential games (thus two simple local information games): a \emph{coordination game} and the \emph{Prisoner's Dilemma}. The analysis of these games highlights that the stationary distribution of the two dynamics can significantly differ. However, it turns out that for both games the Markov chain induced by the all-logit is reversible, just as for the one-logit dynamics. More surprisingly, we see that the expected number of players taking a certain action in each one of these games is exactly the same regardless whether the expectation is taken according the stationary distribution of the all-logit or of the one-logit. Finally, we observe that the mixing time of the all-logit dynamics is asymptotically the same than the mixing time of the one-logit.
Next sections will show that these results are not accidental.

\paragraph{Two-player coordination games.}
These are games in which the players have an advantage in selecting the same strategy.
They are often used to model the spread of a new technology \cite{youngTR00}:
two players have to decide whether to adopt or not a new technology.
Each player prefers to adopt the same technology as the other player.
We denote by $-1$ the strategy of adopting the new technology and by $+1$ the strategy of
adopting the {old} technology.
The game is formally described by the following  payoff matrix
\begin{equation}
\label{eq:coorddef}
\usebox{\tpbox}
\end{equation}
We assume that $a>d$ and $b>c$ (meaning that players prefer to coordinate)
and that $a-d = b-c = \Delta$ (meaning that there is not a risk dominant strategy \cite{flMIT98}).
It is easy to see that this game is a potential game.
It is well known that the stationary distribution of the
one-logit of a potential game is the Gibbs distribution, that assigns to $\x\in S$ probability
$e^{-\beta\Phi(\x)}/Z$, where $Z=\sum_{\x\in S} e^{-\beta\Phi(\x)}$
is the {\em partition function}.

The transition matrix of the Markov chain induced by the all-logit dynamics is
$$
P =
\left(
\begin{array}{c|cccc}
   & -- & -+ & +- & ++\\
\hline
-- & (1-p)^2 & p(1-p) & p(1-p) & p^2\\
-+ & (1-p)p & p^2 & (1-p)^2 & (1-p)p\\
+- & p(1-p) & (1-p)^2 & p^2 & p(1-p)\\
++ & p^2 & p(1-p) & p(1-p) & (1-p)^2
\end{array}
\right)
$$
where $p = 1 / (1+e^{\Delta\beta})$. Observe that this transition matrix
is doubly-stochastic, that implies that the stationary distribution of the all-logit is uniform (and hence very different from the one-logit case).
However, it is easy to check that the chain is reversible and the mixing time is $\Theta\left(e^{\Delta \beta}\right)$ (as in the one-logit case).
Moreover, the expected number of players adopting the new strategy at stationarity is $+1$, both when considering the one- and the all-logit dynamics.

\paragraph{Prisoner's Dilemma.}
The Prisoner's Dilemma game is described by the payoff matrix given in \eqref{eq:coorddef}, where with $-1$ we denote the strategy \texttt{Confess} and with $+1$ the strategy \texttt{Defect}. Moreover, payoffs satisfy the following conditions: (i) $a > d$ (so that $--$ is a Nash equilibrium); (ii) $b < c$ (so that $++$ is not a Nash equilibrium); (iii) $2a < c+d < 2b$ (so that $++$ is the social optimum and $--$ is the worst social profile). It is easy to check that the game is a potential game.

The transition matrix of the Markov chain induced by the all-logit dynamics is
$$
P =
\left(
\begin{array}{c|cccc}
   & -- & -+ & +- & ++\\
\hline
-- & (1-p)^2 & p(1-p) & p(1-p) & p^2\\
-+ & (1-p)(1-q) & p(1-q) & q(1-p) & pq\\
+1 & (1-p)(1-q) & q(1-p) & p(1-q) & pq\\
++ & (1-q)^2 & q(1-q) & q(1-q) & q^2
\end{array}
\right)
$$
where we let $p = 1 / (1+e^{(a-d)\beta})$ be the probability a player does not confess given the other player is currently confessing
and $q = 1/(1+ e^{(c-b) \beta})$ be the probability a player does not confess given the other player is currently not confessing.
Note that both $p$ and $q$ go to $0$ as $\beta$ goes to infinity.

It is easy to check that the transition matrix is reversible (as for the one-logit). The stationary distribution is
$$
\pi(--)=\frac{(1-q)^2}{(1+p-q)^2} \qquad \pi(++)=\frac{p^2}{(1+p-q)^2} \qquad \pi(-+)=\pi(+-)=\frac{p(1-q)}{(1+p-q)^2}.
$$
Moreover, we can see that that the mixing time is upper bounded by a constant independent of $\beta$ (as for the one-logit).
You may also check that the expected number of confessing prisoners is exactly the same in the stationary distribution of the one- and of the all-logit.

\section{Reversibility and stationary distribution}\label{sec::reversibility}
Reversibility is an important property of Markov chains and,
in general, of stochastic processes.
Roughly speaking, for a reversible Markov chain the stationary frequency of transitions from a state $x$ to a state $y$ is equal to the stationary frequency of transitions from $y$ to $x$.
It is easy to see that the one-logit for a game $\mathcal{G}$ are reversible if and only if $\mathcal{G}$ is a potential game.
This does not hold for the all-logit.
Indeed, we will prove that the class of games for which the all-logit are reversible is exactly the class of local interaction games.

\subsection{Reversibility criteria}
As previously stated, a Markov chain $\calM$ is reversible if 
there exists a distribution $\pi$ such that
the detailed balance condition~\eqref{eq:rev_det_bal} is satisfied.
The Kolmogorov reversibility criterion allows us to establish the
reversibility of a process directly from the transition probabilities.
Before stating the criterion, we introduce the following notation.
A \emph{directed path} $\Gamma$ from state $x\in S$ to state $y\in S$
is a sequence of states $\langle x_0, x_1, \ldots, x_\ell\rangle$
such that $x_0=x$ and $x_\ell =y$.
The probability $\Prob{}{\Gamma}$ of path $\Gamma$  is defined
as $\Prob{}{\Gamma} = \prod_{j= 1}^\ell P(x_{j-1},x_j)$.
The \emph{inverse of path}
$\Gamma=\langle x_0,x_1,\ldots,x_\ell\rangle$ is the path
$\Gamma^{-1}=\langle x_\ell, x_{\ell-1},\ldots,x_0\rangle$.
Finally, a cycle $C$ is simply a path from a state $x$ to itself.
We are now ready to state Kolmogorov's reversibility criterion
(see, for example, \cite{Kelly79}).

\begin{theorem}[Kolmogorov's Reversibility Criterion]
An irreducible Markov chain $\calM$ with state space $S$ and transition matrix $P$
is reversible if and only if for every cycle $C$ it holds that
$$
\Prob{}{C}=\Prob{}{C^{-1}}.
$$
\end{theorem}
The following lemma will be very useful for proving reversibility conditions for the all-logit dynamics and for stating a closed expression for its stationary distribution.
\begin{lemma}
\label{lemma:rev_eq}
Let $\calM$ be an irreducible Markov chain with transition probability $P$ and state space $S$.
$\calM$ is reversible if and only if for every pair of states $x,y\in S$,
there exists a constant $c_{x,y}$ such that for all paths $\Gamma$ from $x$ to $y$,
it holds that
$$
\frac{\Prob{}{\Gamma}}{\Prob{}{\Gamma^{-1}}} = c_{x,y}.
$$
\end{lemma}
\begin{proof}
Fix $x,y\in S$ and consider two paths,
$\Gamma_1$ and $\Gamma_2$, from $x$ to $y$. Let $C_1$ and $C_2$ be the cycles
$C_1=\Gamma_1\circ\Gamma_2^{-1}$ and
$C_2=\Gamma_2\circ\Gamma_1^{-1}$, where $\circ$ denotes the concatenation of paths.
If $\calM$ is reversible then, by the Kolmogorov Reversibility Criterion,
$\Prob{}{C_1}=\Prob{}{C_2}.$
On the other hand,
$$
\Prob{}{C_1}=\Prob{}{\Gamma_1}\cdot\Prob{}{\Gamma_2^{-1}}
\quad\text{and}\quad
\Prob{}{C_2}=\Prob{}{\Gamma_2}\cdot\Prob{}{\Gamma_1^{-1}}.
$$
Thus
$$
\frac{\Prob{}{\Gamma_1}}{\Prob{}{\Gamma_1^{-1}}}=
\frac{\Prob{}{\Gamma_2}}{\Prob{}{\Gamma_2^{-1}}}.
$$
For the other direction, fix $z\in S$ and, for all $x\in S$,
set $\tilde\pi(x)=c_{z,x}/Z$, where $Z=\sum_x c_{z,x}$ is the normalizing constant.
Now consider any two states $x,y\in S$ of $\calM$,
let $\Gamma_1$ be any path from $z$ to $x$ and
and set $\Gamma_2=\Gamma_1\circ\langle x,y\rangle$
(that is, $\Gamma_2$ is $\Gamma_1$ concatenated with the edge~$(x,y)$).
We have that
\begin{align*}
    \frac{\tilde\pi(x)}{\tilde\pi(y)} & = \frac{c_{z,x}}{c_{z,y}}\\
                                      & = \frac{\Prob{}{\Gamma_1}}{\Prob{}{\Gamma_1^{-1}}}\cdot
                                        \frac{\Prob{}{\Gamma_2}}{\Prob{}{\Gamma_2^{-1}}} \\
                                      & = \frac{\Prob{}{\Gamma_1}}{\Prob{}{\Gamma_1^{-1}}}\cdot
                         \frac{\Prob{}{\Gamma_1^{-1}}\cdot P(y,x)}{\Prob{}{\Gamma_1}\cdot P(x,y)}\\
                                      & = \frac{P(y,x)}{P(x,y)}
\end{align*}
and therefore $\calM$ is reversible with respect to $\tilde\pi$.
\end{proof}

\subsection{All-logit reversibility implies potential games}
In this section we prove that if the all-logit for a game $\mathcal{G}$ are
reversible then $\mathcal{G}$ is a potential game.

The following lemma shows a condition on the cumulative utility
of a game $\mathcal{G}$ that is necessary and sufficient for the reversibility
of the all-logit of $\mathcal{G}$.
\begin{lemma}
\label{lemma:rev_propU}
The all-logit for game $\mathcal{G}$ are reversible
if and only if the following property holds for every $\x,\y,\z \in S$:
\begin{equation}\label{eq:U_rev}
U(\x,\y) - U(\y,\x) = \Big(U(\x, \z) + U(\z,\y)\Big) - \Big(U(\y, \z) + U(\z,\x)\Big).
\end{equation}
\end{lemma}
\begin{proof}
To prove the only if part,
pick any three $\x,\y,\z\in S$ and consider paths
$\Gamma_1=\langle \x,\y\rangle$
$\Gamma_2=\langle\x,\z,\y\rangle$.
From Lemma~\ref{lemma:rev_eq} we have that reversibility implies
$$
\frac{\Prob{}{\Gamma_1}}{\Prob{}{\Gamma_1^{-1}}}=
\frac{\Prob{}{\Gamma_2}}{\Prob{}{\Gamma_2^{-1}}}
$$
whence
$$
 \frac{e^{\beta U(\x,\y)}}{T(\x)} \frac{T(\y)}{e^{\beta U(\y,\x)}} = \frac{e^{\beta U(\x,\z)}}{T(\x)} \frac{e^{\beta U(\z,\y)}}{T(\z)} \frac{T(\y)}{e^{\beta U(\y,\z)}}  \frac{T(\z)}{e^{\beta U(\z,\x)}}.
$$
which in turn implies~\eqref{eq:U_rev}.

As for the if part,
let us fix state $\z\in S$ and define
$\tilde\pi(\x) = \frac{P(\z,\x)}{Z\cdot P(\x,\z)}$,
where $Z$ is the normalizing constant.
For any $\x,\y\in S$, we have
$$
 \frac{\tilde\pi(\x)}{\tilde\pi(\y)} = \frac{P(\z,\x)}{P(\x,\z)} \cdot \frac{P(\y,\z)}{P(\z,\y)} = \frac{e^{\beta U(\z,\x)}}{e^{\beta U(\x,\z)}} \cdot \frac{e^{\beta U(\y,\z)}}{e^{\beta U(\z,\y)}} \cdot \frac{T(\x)}{T(\y)} = \frac{e^{\beta U(\y,\x)}}{e^{\beta U(\x,\y)}} \cdot \frac{T(\x)}{T(\y)} = \frac{P(\y,\x)}{P(\x,\y)},
$$
where the first equality follows from the definition of $\tilde\pi$, the second and the fourth follow from \eqref{eq:transition} and the third follows from \eqref{eq:U_rev}.
Therefore, the detailed balance equation holds for $\tilde\pi$ and thus the
Markov chain is reversible.
\end{proof}
We are now ready to prove that the all-logit are
reversible only for potential games.
\begin{prop}
\label{prop:onlypot}
If the all-logit for game $\G$ are reversible then $\G$ is a potential game.
\end{prop}
\begin{proof}
We show that if the all-logit are reversible then
the utility improvement $I(\omega)$ over any circuit $\omega$ of length $4$ is $0$.
The theorem then follows by Theorem~\ref{MS96}.

Consider circuit $\omega=\langle \x,\z,\y,\w,\x \rangle$ and
let $i$ be the player in which $\x$ and $\z$ differ and
let $j$ be the player in which $\z$ and $\y$ differ.
Then $\y$ and $\w$ differ in player $i$ and $\w$ and $\x$ differ in player $j$.
In other words, $\z=(\x_{-i},y_i)=(\y_{-j},x_j)$ and
                $\w=(\x_{-i},y_j)=(\y_{-i},x_i)$.
Therefore we have that
$$
\begin{array}{lcl}
U(\x,\y)=\sum_{k \neq i,j} u_k(\x) + u_i(\z) + u_j(\w) & \quad &
U(\y,\x)=\sum_{k \neq i,j} u_k(\y) + u_i(\w) + u_j(\z)\\
U(\x,\z)=\sum_{k \neq i,j} u_k(\x) + u_i(\z) + u_j(\x) & \quad &
U(\z,\y)=\sum_{k \neq i,j} u_k(\z) + u_i(\z) + u_j(\y)\\
U(\y,\z)=\sum_{k \neq i,j} u_k(\y) + u_i(\y) + u_j(\z) & \quad &
U(\z,\x)=\sum_{k \neq i,j} u_k(\z) + u_i(\x) + u_j(\z)
\end{array}
$$
By plugging the above expressions into~\eqref{eq:U_rev} and
rearranging terms, we obtain
$$
\Big(u_i(\z) - u_i(\x)\Big)
+
\Big(u_j(\y) - u_j(\z)\Big)
+
\Big(u_i(\w) - u_i(\y)\Big)
+
\Big(u_j(\x) - u_j(\w)\Big)
=0
$$
which shows $I(\omega)=0$.
\end{proof}

\subsection{A necessary and sufficient condition for all-logit reversibility}
In the previous section we have established that
the all-logit are reversible only for potential games and therefore,
from now on,
we only consider potential games $\G$  with potential function $\Phi$.
In this section we present in Proposition~\ref{prop:DiffPot}
a necessary and sufficient condition for reversibility that involves
only the potential function.
The condition will then be used in the next section
to prove that
local interaction
games are exactly the games whose
all-logit are reversible.

\begin{prop}
\label{prop:DiffPot}
The all-logit for a game $\G$ with potential $\Phi$ are reversible
if and only if, for all strategy profiles $\x,\y \in S$,
\begin{equation}
 \label{eq:psi_cond}
 K(\x,\y) = K(\y, \x),
\end{equation}
where $K$ is as defined in~\eqref{eq:kappa}.
\end{prop}
\begin{proof}
If $K(\x, \y) = K(\y, \x)$, then
$$
 \sum_i \Big(\Phi(\y_{-i},x_i) - \Phi(\y)\Big) - \sum_i \Big(\Phi(\x_{-i},y_i) - \Phi(\x)\Big) = 2 \Big(\Phi(\x) - \Phi(\y)\Big).
$$
Hence, for any pair of strategy profiles $\x, \y$ we have
\begin{align*}
 U(\x,\y) - U(\y,\x) & = n \Big(\Phi(\x) - \Phi(\y)\Big) + \sum_i \Big(u_i(\x_{-i},y_i) - u_i(\x)\Big) - \sum_i \Big(u_i(\y_{-i},x_i) - u_i(\y)\Big)\\
 & = n \Big(\Phi(\x) - \Phi(\y)\Big) + \sum_i \Big(\Phi(\y_{-i},x_i) - \Phi(\y)\Big) - \sum_i \Big(\Phi(\x_{-i},y_i) - \Phi(\x)\Big)\\
 & = (n + 2) \Big(\Phi(\x) - \Phi(\y)\Big).
\end{align*}
It is then immediate to check that \eqref{eq:U_rev} holds.

As for the other direction,
we proceed by induction on the Hamming distance between $\x$ and $\y$.
Let $\x$ and $\y$ be two profiles at Hamming distance $1$;
that is, $\x$ and $\y$ differ in only one player, say $j$.
This implies that
$(y_j,\x_{-j})=\y$ and $(x_j,\y_{-j})=\x$.
Moreover, for $i\ne j$, $(y_i,\x_{-i})=\x$ and
                        $(x_i,\y_{-i})=\y$.
Thus,
\begin{align*}
 & K(\x,\y) - K(\y,\x) = \sum_{i} \Big(\Phi(y_i, \x_{-i}) - \Phi(x_i, \y_{-i})\Big) - (n - 2) \Big(\Phi(\x) - \Phi(\y)\Big)\\
 & \qquad = \Big(\Phi(y_j, \x_{-j}) - \Phi(x_j, \y_{-j})\Big) + \sum_{i \neq j} \Big(\Phi(y_i, \x_{-i}) - \Phi(x_i, \y_{-i})\Big) - (n - 2) \Big(\Phi(\x) - \Phi(\y)\Big)\\
 & \qquad = \Big(\Phi(\y) - \Phi(\x)\Big) + (n - 1) \Big(\Phi(\x) - \Phi(\y)\Big) - (n - 2) \Big(\Phi(\x) - \Phi(\y)\Big) = 0.
\end{align*}
Now assume that the claim holds for any pair of profiles at Hamming distance $k<n$ and
let $\x$ and $\y$ be two profiles at distance $k+1$.
Let $j$ be any player such that $x_j \neq y_j$ and let $\z = (y_j,\x_{-j})$:
$\z$ is at distance at most $k$ from $\x$ and from $\y$.
Consider paths $\Gamma_1=\langle \x,\y\rangle$ and $\Gamma_2=\langle\x,\z,\y\rangle$.
From Lemma~\ref{lemma:rev_eq} we have that reversibility implies
$$
 \frac{e^{\beta K(\x,\y)}}{\gamma_A(\x)} \frac{\gamma_A(\y)}{e^{\beta K(\y,\x)}} = \frac{e^{\beta K(\x,\z)}}{\gamma_A(\x)} \frac{e^{\beta K(\z,\y)}}{\gamma_A(\z)} \frac{\gamma_A(\y)}{e^{\beta K(\y,\z)}} \frac{\gamma_A(\z)}{e^{\beta K(\z,\x)}}.
$$
Hence $K(\x,\y) - K(\y,\x) = \Big(K(\x, \z) - K(\z,\x)\Big) + \Big(K(\y, \z) - K(\z,\y)\Big)$ and the thesis follows from the inductive hypothesis.
\end{proof}

\subsection{Reversibility and local interaction games}
Here we prove that the games whose all-logit are reversible
are exactly
the local interaction games.

A potential
$\Phi:S_1\times\cdots\times S_n\rightarrow\mathbb{R}$
is a {\em two-player potential} if there exist $u,v\in[n]$
such that, for any $\x,\y\in S$ with  $x_u=y_u$ and $x_v=y_v$ we have
$\Phi(\x)=\Phi(\y)$.
In other words, $\Phi$ is a function of only its $u$-th and $v$-th argument.
An interesting fact about two-player potential games is given by the following lemma.
\begin{lemma}
 \label{lemma:2player}
 Any two-player potential satisfies~\eqref{eq:psi_cond}.
\end{lemma}
\begin{proof}
Let $\Phi$ be a two-player potential and let $u$ and $v$ be its two
players. Then we have that for $w\ne u,v$,
$\Phi(y_w,\x_{-w})=\Phi(\x)$ and that
$\Phi(y_u,\x_{-u})=\Phi(x_v,\y_{-v})$ and
$\Phi(y_v,\x_{-v})=\Phi(x_u,\y_{-u})$.
Thus
$$
 K(\x,\y)=\Phi(y_u,\x_{-u})+\Phi(y_v,\x_{-v})
$$
and
\[
 K(\y,\x) =\Phi(x_v,\y_{-v})+\Phi(x_u,\y_{-u}) =\Phi(y_u,\x_{-u})+\Phi(y_v,\x_{-v}).\qedhere
\]
\end{proof}

We say that
a potential $\Phi$ is the sum of two-player potentials
if there exist $N$ two-player potentials
$\Phi_1,\ldots,\Phi_N$ such that
$\Phi=\Phi_1+\cdots+\Phi_N$.
It is easy to see that generality is not lost by further requiring that $1 \leq l\ne l' \leq N$ implies
$(u_l,v_l)\ne (u_{l'},v_{l'})$, where $u_l$ and $v_l$ are the two players
of potential $\Phi_l$.
At every game $\G$ whose potential is the sum of two-player potentials, i.e., $\Phi=\Phi_1+\cdots+\Phi_N$,
we can associate a {\em social graph} $G$ that has a vertex for each player of $\G$
and has edge~$(u,v)$ iff there exists $l$ such that
potential $\Phi_l$ depends on players $u$ and $v$.
In other words, each game whose potential is the sum of two-player potentials is a local interaction
game.

Observe that the sum of two potentials satisfying~\eqref{eq:psi_cond} also satisfies~\eqref{eq:psi_cond}.
Hence we have the following proposition.
\begin{prop}
\label{prop:soc_rev}
The all-logit dynamics for a local interaction game are reversible.
\end{prop}

Next we prove that if an $n$-player potential $\Phi$ satisfies~\eqref{eq:psi_cond} then it can be written as the
sum of at most $N=\binom{n}{2}$ two-player potentials,
$\Phi_1,\ldots,\Phi_N$ and thus it represents a local interaction game.
We do so by describing an effective procedure that constructs the
$N$ two-player potentials.

Let us fix a strategy $w^\star_i$ for each player $i$ and denote as $\w^\star$ the strategy profile $(w^\star_1, \ldots, w^\star_n)$.
Moreover, we fix an arbitrary ordering
$(u_1, v_1),\ldots,(u_N,v_N)$ of the $N$ unordered pairs of players.
For a potential $\Phi$ we define
the sequence $\vartheta_0,\ldots,\vartheta_{N}$
of potentials as follows:
$\vartheta_0=\Phi$ and, for $i=1,\ldots,N$, set
\begin{equation}
\label{eq:fi}
\vartheta_i=\vartheta_{i-1}-\Phi_{i}
\end{equation}
where, for $\x\in S$, $\Phi_i(\x)$ is defined as
$$
\Phi_{i}(\x)=\vartheta_{i-1}(x_{u_i},x_{v_i},\w^\star_{-u_i v_i}).
$$
Observe that, for $i=1,\ldots,N$, $\Phi_{i}$ is a
two-player potential and its players are  $u_i$ and $v_i$.
From Lemma~\ref{lemma:2player}, $\Phi_{i}$ satisfies~\eqref{eq:psi_cond}.
Hence,
if $\Phi$ satisfies~\eqref{eq:psi_cond}, then also $\vartheta_i$, for $i=1,\ldots,N$, satisfies~\eqref{eq:psi_cond}.

By summing for $i=1,\ldots,N$ in~\eqref{eq:fi} we obtain
$$
\sum_{i=1}^N\vartheta_i= \sum_{i=0}^{N-1} \vartheta_i-\sum_{i=1}^{N} \Phi_i.
$$
Thus
$$
\Phi-\vartheta_N=\sum_{i=1}^N \Phi_i.
$$
The next two lemmas prove that, if $\Phi$
satisfies~\eqref{eq:psi_cond},
then $\vartheta_N$ is identically zero. This implies that $\Phi$ is the
sum of at most $N$ non-zero two-player potentials and thus a
local interaction game.

A \emph{ball} $B(r,\x)$ of radius $r \leq n$ centered in $\x \in S$
is the subset of $S$ containing all profiles $\y$
that differ from $\x$ in at most $r$ coordinates.
\begin{lemma}
\label{lemma:algo}
For any $n$-player potential function $\Phi$
and for any ordering of the pairs of players,
$\vartheta_N(\x)=0$ for every $\x\in B(2,\w^\star)$.
\end{lemma}
\begin{proof}
We distinguish three cases based on the distance of $\x$ from $\w^\star$.\\
\underline{$\x=\w^\star$:} for every $i \geq 1$, we have
$$
\vartheta_i(\w^\star)=\vartheta_{i-1}(\w^\star)-\Phi_{i}(\w^\star)=\vartheta_{i-1}(\w^\star)-\vartheta_{i - 1}(\w^\star)
= 0.
$$
\underline{$\x$ is at distance $1$ from $\w^\star$:}
That is, there exists $u\in[n]$ such that $\x=(x_u,\w^\star_{-u})$, with $x_u \neq w_u^\star$.
Let us denote by $t(u)$ the smallest $t$ such that
the $t$-th pair contains $u$. We next show that
for $i \geq t(u)$, $\vartheta_i(\x) = 0$.
Indeed, we have that if
$u$ is a component of the $i$-th pair then
$$
\vartheta_i(\x)=\vartheta_{i-1}(\x)-\Phi_{i}(\x)=\vartheta_{i-1}(\x)-
        \vartheta_{i-1}(\x) = 0;
 $$
On the other hand,
if $u$ is not a component of the $i$-th pair then
$$
\vartheta_i(\x)=\vartheta_{i-1}(\x)-\Phi_{i}(\x)=
            \vartheta_{i-1}(\x)-\vartheta_{i-1}(\w^\star)=\vartheta_{i-1}(\x);
 $$
\underline{$\x$ is at distance $2$ from $\w^\star$:}
That is, there exist $u$ and $v$ such that
$\x=(x_u,x_v,\w^\star_{-uv})$, with $x_u \neq w^\star_u$ and $x_v \neq w^\star_v$.
Let $t$ be the index of the pair $(u,v)$. Notice that $t\geq t(u),t(v)$.
We show that $\vartheta_t(\x)=0$ and that this value
does not change for all $i>t$.
Indeed, we have
 $$
\vartheta_t(\x)=\vartheta_{t-1}(\x)-\Phi_{t}(\x)
               =\vartheta_{t-1}(\x)-\vartheta_{t-1}(\x) = 0;
 $$
If instead neither of $u$ and $v$ belongs to the $i$-th pair, with $i>t$,
then we have
$$
\vartheta_i(\x)=\vartheta_{i-1}(\x)-\Phi_{i}(\x)
               =\vartheta_{i-1}(\x)-\vartheta_{i - 1}(\w^\star)
               =\vartheta_{i-1}(\x);
 $$
Finally,
suppose that the $i$-th pair, for $i>t$, contains exactly one of $u$ and
$v$, say $u$.
Then we have
 $$
\vartheta_i(\x)=\vartheta_{i-1}(\x)-\Phi_{i}(\x)
               =\vartheta_{i-1}(\x)-\vartheta_{i-1}(x_u,\w^\star_{-u}).
 $$
We conclude the proof by observing that
$t(u) \leq t \leq i-1$ and thus, by the previous case,
$\vartheta_{i-1}(x_u,\w^\star_{-u})=0$.
\end{proof}

The next lemma shows that if a potential $\vartheta$ satisfies~\eqref{eq:psi_cond} and is constant in a ball of
radius $2$, then it is constant everywhere.
\begin{lemma}
\label{lemma:all_zero}
Let $\vartheta$ be a function that satisfies \eqref{eq:psi_cond}.
If there exist $\x\in S$ and $c\in\mathbb{R}$ such that $\vartheta(\y) = c$
for every $\y\in B(2,\x)$, then $\vartheta(\y) = c$ for every $\y \in S$.
\end{lemma}
\begin{proof}
Fix $h>2$ and suppose that $\vartheta(\z)=c$ for every $\z\in B(h-1,\x)$.
Consider $\y\in B(h,\x)\setminus B(h-1,\x)$ and observe that
$(y_i,\x_{-i}) \in B(h-1,\x)$ and $(x_i, \y_{-i}) \in B(h-1,\x)$ for every $i$ such that $x_i \neq y_i$.
Then, since $\vartheta$ satisfies \eqref{eq:psi_cond}, we have
$$
(h - 2) \left(\vartheta(\x) - \vartheta(\y)\right) = \sum_{i \colon x_i \neq y_i} \Big(\vartheta(y_i, \x_{-i}) - \vartheta(x_i, \y_{-i})\Big) = 0,
$$
that implies $\vartheta(\y) = \vartheta(\x) = c$.
\end{proof}
We can thus conclude that if the all-logit of a potential game $\G$ are reversible then
$\G$ is a local interaction game.
By combining this result with Proposition~\ref{prop:onlypot} and
Proposition~\ref{prop:soc_rev}, we obtain
\begin{theorem}
The all-logit dynamics of game $\G$ are reversible if and only if $\G$ is a local interaction game.
\end{theorem}

As a corollary of this theorem we have a closed form for the stationary distribution of the all-logit for local interaction games.
\begin{cor}[Stationary distribution]
\label{cor:stationary}
Let $\G$ be a local interaction game with potential function $\Phi$.
Then the stationary distribution of the all-logit for $\G$ is
\begin{equation}
\label{eq:stationary}
\pi_A(\x) \propto \sum_{\y \in S} e^{-\beta K(\x,\y)}.
\end{equation}
\end{cor}
\begin{proof}
Fix any profile $\y$.
The detailed balance equation gives for every $\x \in S$
$$
\frac{\pi_A(\x)}{\pi_A(\y)} =
\frac{P(\y,\x)}{P(\x,\y)} =
    e^{\beta(K(\x,\y) - K(\y,\x))} \frac{\gamma_A(\x)}{\gamma_A(\y)}.
$$
By Proposition~\ref{prop:DiffPot} we have
$$
\pi_A(\x) = \gamma_A(\x) \cdot \frac{\pi_A(\y)}{\gamma_A(\y)}.
 $$
 Since the term $\frac{\pi_A(\y)}{\gamma_A(\y)}$ is constant for each profile $\x$, the claim follows.
\end{proof}

Note that for a local interaction game $\G$ with potential function $\Phi$,
we write $\pi_1(\x)$, the stationary distribution of the
one-logit of $\G$, as $\pi_1(\x)=\gamma_1(\x)/Z_1$ where
$\gamma_1(\x)=e^{-\beta\Phi(\x)}$ is the Boltzmann factor
and $Z_1=\sum_\x\gamma_1(\x)$ is the partition function.
From Corollary~\ref{cor:stationary}, we derive that
$\pi_A(\x)$,
the stationary distribution of the all-logit of $\G$,
can be written in similar way; 
that is, $\pi_A(\x)=\frac{\gamma_A(\x)}{Z_A}$,
where $\gamma_A(\x)=\sum_\y e^{-\beta K(\x,\y)}$ 
and 
$$Z_A = \sum_{\x \in S} \gamma_A(\x)=\sum_{\x,\y\in S} e^{-\beta K(\x,\y)}.$$ 
The $Z_A$ factor can thus be considered as the partition function 
of the all-logit.

\section{Observables of local information games}\label{sec:observables}
In this section we study observables of local interaction games and
we focus on the relation between the expected value~$\mobs{O}{\pi_1}$ of an observable $O$
at the stationarity of the one-logit
and its expected value~$\mobs{O}{\pi_A}$ at the stationarity of the all-logit dynamics.
We start by studying invariant observables, that is,
observables for which the two expected values coincide.
In Theorem~\ref{thm:suff}, we give a sufficient condition for an
observable to be invariant.
The sufficient condition is related
to the existence of a {\em decomposition}
of the set $S \times S$ that
decomposes the quantity $K$ appearing in the
expression for the stationary distribution of the all-logit of the
local interaction game $\G$ (see Eq.~\ref{eq:stationary})
into a sum of two potentials.
In Theorem~\ref{thm:suff} we show that if $\G$ admits
such a decomposition $\mu$ and in addition
observable $O$ is also decomposed by $\mu$
(see Definition~\ref{def:obs}) then
$O$ has the same expected value at the stationarity of the
one-logit and of the all-logit.
We then go on to show that all local interaction games on
{\em bipartite} social graphs admit a decomposition permutation
(see Theorem~\ref{thm:mono_iff_bis}) and give examples
of invariant observables.

We then look at local interaction games $\G$ on general social graphs $G$
and show that the expected values of a
decomposable observable $O$ with respect to the stationary
distributions of the one-logit and of the all-logit
differ by a quantity that depends on $\beta$ and on how far away the social
graph $G$ is from being bipartite (which in turn is related
to the smallest eigenvalue of $G$ \cite{trevisanSTOC09}).

The above findings follow from a
relation between the partition functions of the one-logit and of the
all-logit that might be of independent interest.
More precisely, in Theorem~\ref{thm:bipdec}
we show that if the game $\G$ admits a decomposition
then the partition function of the
all-logit is the square of the partition function of the one-logit.
The partition function of the one-logit is easily seen to be equal to the
partition function of the canonical ensemble used
in Statistical Mechanics (see for example~\cite{landau}).
It is well known that a partition function of a
canonical ensemble that is the union of two independent
canonical ensembles is the product of the two partition functions.
Thus Theorem~\ref{thm:bipdec} (and Corollary~\ref{cor:station_bip}) can be seen as a further confirmation
that the all-logit can be decomposed into two independent one-logit
dynamics.

\subsection{Decomposable observables for bipartite social graphs}
\label{sec:dec}
We start by introducing the concept of a {\em decomposition} and we prove that
for all local interaction games on a bipartite social graph there exists
a decomposition.
Then we define the concept of a {\em decomposable} observable and prove that a decomposable observable has the same expectation
at stationarity for the one-logit and the all-logit.

\begin{definition}
\label{def:dec}
A permutation $$\mu \colon (\x,\y)\mapsto(\mu_1(\x,\y),\mu_2(\x,\y))$$
    of $S\times S$ is a
    {\em decomposition} for a local interaction
    game $\G$ with potential $\Phi$ if,
    for all $(\x,\y)$, we have that
        $$K(\x,\y)=\Phi(\mu_1(\x,\y))+\Phi(\mu_2(\x,\y)),$$
$\mu_1(\x,\y)=\mu_2(\y,\x)$ and $\mu_2(\x,\y)=\mu_1(\y,\x)$.
\end{definition}
Observe that if $\mu$ decomposes local interaction game $\G$ then 
$$\pi_A(\x)=\sum_\y \pi_1(\mu_1(\x,\y))\cdot\pi_1(\mu_2(\x,\y))$$

We first show a relation between the partition functions of the one-logit and of the
all-logit that might be of independent interest.
\begin{theorem}
\label{thm:bipdec}
If a local interaction game $\G$ admits a decomposition  $\mu$
then $Z_A=Z_1^2$.
\end{theorem}
\begin{proof}
From \eqref{eq:stationary} and from the fact that
$\mu$ is a permutation of $S\times S$,
we have
\[
        Z_A=\sum_{\x,\y} e^{-\beta K(\x,\y)}
           =\sum_{\x,\y} e^{-\beta [\Phi(\mu_1(\x,\y))+
                                      \Phi(\mu_2(\x,\y))]}
           =\sum_{\x,\y} e^{-\beta[\Phi(\x)+\Phi(\y)]}
           = Z_1^2.\qedhere
\]
\end{proof}

We next prove that for all local interaction games on a bipartite social graph there exists a decomposition.
We start by showing that we can decompose $K(\x,\y)$ in the contributions of each edge
of the social graph $G$ of the local interaction game $\G$.
Specifically, for strategy profiles $\x$ and $\y$
and edge $e=(u,v)$ of $G$ we define $K_e(\x,\y)$ as
\begin{equation}
\label{eq:ke}
 K_e(\x,\y) = \Phi_e(x_u,y_v) + \Phi_e(y_u,x_v).
\end{equation}
Then we have the following lemma that will be useful for giving a sufficient condition for having a decomposition.
\begin{lemma}
\label{lemma:K}
$$K(\x,\y) = \sum_e K_e(\x,\y).$$
\end{lemma}
\begin{proof}
By definition $K(\x,\y)=\sum_i \Phi(\x_{-i},y_i)-(n-2)\Phi(\x)$. Then 
by expressing $\Phi$ as sum of the potential over the edges we have
$$K(\x,\y)=\sum_i\sum_e\Phi_e(\x_{-i},y_i)-(n-2)\Phi(\x).$$
Then observe that edge $e=(u,v)$ and each of the $(n-2)$ vertices 
$i\ne u,v$ contribute
$\Phi_e(x_u,x_v)$ to the sum. 
On the other hand, 
the total contribution for $e=(u,v)$ and $i=u,v$
is $K_e(\x,\y)=\Phi_e(y_u,x_v) + \Phi_e(x_u,y_v)$.
Therefore we obtain 
\begin{align*}
K(\x,\y)&=\sum_i\sum_{e=(u,v)}\Phi_e(\x_{-i},y_i)-(n-2)\Phi(\x)\\
        &=\sum_{e=(u,v)}\left[ (n-2) \Phi_e(x_u,x_v) + K_e(\x,\y)\right]
                -(n-2)\Phi(\x)\\
        &=\sum_{e}K_e(\x,\y).\qedhere
\end{align*}
%
\end{proof}

From Lemma~\ref{lemma:K}, we then achieve the following sufficient condition for a permutation to be a decomposition.
\begin{lemma}
\label{lemma:condition_kappa}
Let $\G$ be a local interaction game with potential $\Phi$ on a graph $G$.
Consider a permutation $(\x,\y)\mapsto(\tilde\x,\tilde\y)$ such that
for all $\x,\y\in S$ and for each edge $e=(u,v)$ of $G$
at least one of the following equalities holds
\begin{equation}
\label{eq:one}
(\tilde{x}_u,\tilde{x}_v,\tilde{y}_u,\tilde{y}_v)=(x_u,y_v,y_u,x_v), 
\end{equation}
\begin{equation}
\label{eq:two}
(\tilde{x}_u,\tilde{x}_v,\tilde{y}_u,\tilde{y}_v)=(y_u,x_v,x_u,y_v).
\end{equation}
Then, 
$K(\x,\y)=\phi(\tilde\x)+\phi(\tilde\y)$.
\end{lemma}
\begin{proof}
Observe that, 
if one of Equation~\eqref{eq:one} and \eqref{eq:two} holds,
$$\Phi_e(\tx_u,\tx_v)+ \Phi_e(\ty_u,\ty_v)= \Phi_e(x_u,y_v)+\Phi_e(y_u,x_v)=K_e(\x,\y).$$
The corollary then follows by summing over all edges $e$.
\end{proof}

We are now ready for the main result of the section.
 \begin{theorem}
\label{thm:mono_iff_bis}
Let $\G$ be a local interaction game on a bipartite graph $G$.
Then $\G$ admits a decomposition.
\end{theorem}
\begin{proof}
Let $(L,R)$ be the sets of vertices in which $G$ is bipartite.
For each $(\x,\y)\in S\times S$ define 
\begin{equation}
 \label{eq:decomp}
 \txx=\mu_1(\x,\y)=(\x_L, \y_R) \qquad \text{and} \qquad \tyy = \mu_2(\x,\y) = (\y_L,\x_R).
\end{equation}

First of all,
observe that the mapping is an involution and thus it is also a permutation
and that
$\mu_1(\x,\y)=\mu_2(\y,\x)$ and $\mu_2(\x,\y)=\mu_1(\y,\x)$.
Since $G$ is bipartite, for every edge~$(u,v)$
exactly one endpoint is in $L$
and exactly one is in $R$.
If $u \in L$, then we have that
$(\tx_u,\tx_v,\ty_u,\ty_v)=(x_u,y_v,y_u,x_v)$
and thus \eqref{eq:one} is satisfied.
If instead $u \in R$, then we have that
$(\tx_u,\tx_v,\ty_u,\ty_v)=(y_u,x_v,x_u,y_v)$
and thus \eqref{eq:two} is satisfied.
Therefore for each edge one of \eqref{eq:one} and \eqref{eq:two} is satisfied.
By Lemma~\ref{lemma:condition_kappa}, we can conclude that the mapping is a decomposition.
\end{proof}

Consider a local interaction game on a bipartite graph $G = (L, R, E)$ and
let us denote by $L(\x)$ (respectively, $R(\x)$) 
the set of profiles agreeing with $\x$ for every 
vertex of $L$ (respectively, of $R$).
That is, 
$L(\x)=\{\y\colon\y_L=\x_L\}$ and 
$R(\x)=\{\y\colon\y_R=\x_R\}$.
The following corollary of Theorem~\ref{thm:bipdec} and Theorem~\ref{thm:mono_iff_bis}
proves an interesting characterization of the stationary distribution
of the all-logit dynamics for local interaction games on bipartite graphs
that might be of independent interest.
\begin{cor}
\label{cor:station_bip}
 Let $\G$ be a local interaction game on a bipartite graph $G = (L, R, E)$.
 Then for each profile $\x$ we have
 $$
  \pi_A(\x) = \pi_1(L(\x)) \cdot \pi_1(R(\x)).
 $$
\end{cor}
\begin{proof}
 Observe that
 $$
  \pi_1(L(\x)) \cdot \pi_1(R(\x)) = \frac{1}{Z_1^2} \sum_{(\y, \z) \in L(\x) \times R(\x)} e^{-\beta \left(\Phi(\y) + \Phi(\z)\right)}.
 $$
 For each pair $(\y, \z) \in L(\x) \times R(\x)$, consider the profile $\w_{\y,\z} = (\z_L, \y_R)$.
 Then, $(\y, \z) = \mu(\x, \w_{\y,\z})$, where $\mu$ is the decomposition~\eqref{eq:decomp}.
 Note that the correspondence between pairs $(\y, \z) \in L(\x) \times R(\x)$ and profiles $\w$ is actually a bijection.
 Indeed, for each profile $\w \in S$, the pair $(\mu_1(\x,\w),\mu_2(\x,\w))$ belongs to $L(\x) \times R(\x)$.
 Hence and from Theorem~\ref{thm:bipdec}, it follows that
 \[
 \pi_1(L(\x)) \cdot \pi_1(R(\x)) = \frac{1}{Z_A} \sum_{\w} e^{-\beta \left(\Phi(\mu_1(\x,\w)) + \Phi(\mu_2(\x,\w))\right)} = \frac{1}{Z_A} \sum_{\w} e^{-\beta K(\x,\w)} = \pi_A(\x).\qedhere
 \]
\end{proof}
We now define the concept of a decomposable observable.
\begin{definition}
\label{def:obs}
An observable $O$ is {\em decomposable} for local interaction game 
$\G$ if there exists a decomposition $\mu$ of $\G$ such that,
    for all $(\x,\y)$, we have that
        $$O(\x)+O(\y)=O(\mu_1(\x,\y))+O(\mu_2(\x,\y)).$$
\end{definition}
We next prove that a decomposable observable has the same expectation
at stationarity of the one-logit and the all-logit.
\begin{theorem}
\label{thm:suff}
If observable $O$ is decomposable then
$$\langle O,\pi_1\rangle=\langle O,\pi_A\rangle.$$
\end{theorem}
\begin{proof}
Suppose that $O$ is decomposed by $\mu$.
Then we have that, for all $\x\in S$,
$\pi_A(\x)=\sum_\y\pi_1(\mu_1(\x,\y))\cdot\pi_1(\mu_2(\x,\y))$ and
thus
$$
\begin{aligned}
\mobs{O}{\pi_A}&=\sum_\x O(\x)\cdot\pi_A(\x)\\
               &=\sum_{\x,\y} O(\x)\cdot
                        \pi_1(\mu_1(\x,\y))\cdot\pi_1(\mu_2(\x,\y))\\
               &=\frac{1}{2}\sum_{\x,\y}
        \left[O(\x)+O(\y)\right] \cdot\pi_1(\mu_1(\x,\y))\cdot
                                      \pi_1(\mu_2(\x,\y))\\
\end{aligned}
$$
In the last equality we have used that
$\mu_1(\x,\y)=\mu_2(\y,\x)$ and
$\mu_2(\x,\y)=\mu_1(\y,\x)$ which implies that
$$
    \sum_{\x,\y} O(\x)\cdot
        \pi_1(\mu_1(\x,\y))\cdot
        \pi_1(\mu_2(\x,\y))
=
    \sum_{\x,\y} O(\y)\cdot
        \pi_1(\mu_1(\x,\y))\cdot
        \pi_1(\mu_2(\x,\y)).
$$
Now, since $O$ is decomposable we have that
$O(\x)+O(\y)=O(\mu_1(\x,\y))+O(\mu_2(\x,\y))$ and thus we can write
\begin{align*}
\mobs{O}{\pi_A}&=
\frac{1}{2}\sum_{\x,\y}
        \left[O(\mu_1(\x,\y))+O(\mu_2(\x,\y))\right]
                  \cdot\pi_1(\mu_1(\x,\y))\cdot\pi_1(\mu_2(\x,\y))\\
               &=
\frac{1}{2}\sum_{\x,\y}
        \left[O(\x)+O(\y)\right]
                \cdot\pi_1(\x)\cdot\pi_1(\y)\\
               &=
\sum_{\x,\y}
        O(\x)\cdot\pi_1(\x)\cdot\pi_1(\y)\\
               &=
\sum_{\x}
        O(\x)\cdot\pi_1(\x)\cdot\sum_\y \pi_1(\y)\\
              &=\mobs{O}{\pi_1}.\qedhere
\end{align*}
\end{proof}
We now give examples of decomposable observables.
\paragraph{The $\Diff$ observable.}
Let us consider the case that players have only two available strategies, namely $-1$ and $+1$.
We consider the observable $\Diff$ that
returns the (signed) difference between
the number of vertices adopting the strategy $-1$ and
the number of vertices adopting strategy $+1$.  That is, $\Diff(\x)=\sum_u x_u$.
In local interaction games used to model the diffusion of innovation in social networks and the spread
of new technology (see, for example,~\cite{youngTR00}), this observable is a measure of how wide is the adoption of the innovation.
The $\Diff$ observable is also meaningful in the
Ising model for ferromagnetism (see, for example,~\cite{Mart1999}) as it
is the measured magnetism.

To prove that $\Diff$ is decomposable we consider the mapping~\eqref{eq:decomp} and
observe that, for every vertex $u$ and for every $(\x,\y)\in S\times S$,
we have $x_u+y_u=\tx_u+\ty_u$. Whence we conclude that $\Diff(\x)+\Diff(\y)=\Diff(\txx)+\Diff(\tyy)$.

\paragraph{The $\MonoC$ observable.}
Another interesting decomposable observable is the signed difference
$\MonoC$ between the number of ``$-1$''-monochromatic edges of the social graph
(that is, edges in which both endpoints play $-1$) and the
number of ``$+1$''-monochromatic edges.
That is, $\MonoC(\x)=\frac{1}{2} \sum_{(u,v)\in E} \left(x_u+x_v\right)$.
Again, we consider the mapping~\eqref{eq:decomp} and the
decomposability of $\MonoC$ follows from the property that, for every $(\x,\y)\in S\times S$,
we have $x_u+y_u=\tx_u+\ty_u$.

\begin{cor}
Observables $\Diff$ and $\MonoC$ are decomposable and thus, for local interaction games on bipartite social graphs,
$$\mobs{\Diff}{\pi_1}=
  \mobs{\Diff}{\pi_A}
\qquad{\rm and}\qquad
 \mobs{\MonoC}{\pi_1}=
  \mobs{\MonoC}{\pi_A}.$$
\end{cor}

\subsection{General graphs}
Let us start by slightly generalizing concepts of decomposition and decomposable observable.
\begin{definition}
A permutation $$\mu\colon (\x,\y)\mapsto(\mu_1(\x,\y),\mu_2(\x,\y))$$
    of $S\times S$ is an
    {\em $\alpha$-decomposition} for a local interaction
    game $\G$ with potential $\Phi$ if,
    for all $(\x,\y)$, we have that
        $$\left|K(\x,\y)-\Phi(\mu_1(\x,\y))-\Phi(\mu_2(\x,\y))\right| \leq \alpha,$$
$\mu_1(\x,\y)=\mu_2(\y,\x)$ and $\mu_2(\x,\y)=\mu_1(\y,\x)$.
\end{definition}
Note that a decomposition is actually a $0$-decomposition
(see Definition~\ref{def:dec}).

\begin{definition}
An observable $O$ is {\em $\alpha$-decomposable} if it is decomposed by an
$\alpha$-decomposition.
\end{definition}

We prove that for all local interaction games there exists
an $\alpha$-decomposition with $\alpha$ depending only on how far away the social
graph $G$ is from being bipartite.
Specifically, for each edge~$e$ of the social graph we define the \emph{weight} $w_e = \max_{\x, \y \in \G_e} \left(\Phi_e(\x) - \Phi_e(\y)\right)$, i.e., $w_e$ is the maximum difference in the potential $\Phi_e$ of the two-player game $\G_e$.
We say that a subset of edges of $G$ is \emph{bipartiting} if the removal of these edges makes the graph bipartite. We will denote with $B(G)$ the bipartiting subset of minimum weight and with $b(G)$ its weight. We have then the following theorem.
 \begin{theorem}
Let $\G$ be a social interaction game on a graph $G$.
Then, for any $\alpha \geq 2 \cdot b(G)$, $\G$ admits an $\alpha$-decomposition.
\end{theorem}
\begin{proof}
 Let us name as $G'=(L,R,E')$ the bipartite graph obtained by deleting from $G$ the edges of $B(G)$ and consider the mapping~\eqref{eq:decomp}. We know this mapping is actually a permutation and $\mu_1(\x,\y)=\mu_2(\y,\x)$ and $\mu_2(\x,\y)=\mu_1(\y,\x)$. We will show that, for every $\x,\y$
 \begin{equation}
 \label{eq:K_generalG}
  | K(\x,\y) - \Phi(\txx) - \Phi(\tyy) | \leq 2 \cdot b(G),
 \end{equation}
 where $\txx = \mu_1(\x,\y)$ and $\tyy = \mu_2(\x,\y)$.

 Observe that $K(\x,\y) = \sum_{e \in E'} K_e(\x,\y) + \sum_{e \in E \setminus E'} K_e(\x,\y)$. From Theorem~\ref{thm:mono_iff_bis}, for each edge~$e=(u,v) \in E'$ we have $K_e(\x,\y) = \Phi_e(\tx_u,\tx_v) + \Phi_e(\ty_u,\ty_v)$. As for each edge~$e = (u,v) \in E \setminus E'$ we have that the endpoints are either both in $L$ or both in $R$. In both cases, it turns out that
 $$
  \Phi_e(\tx_u,\tx_v) + \Phi_e(\ty_u,\ty_v) = \Phi_e(x_u,x_v) + \Phi_e(y_u,y_v).
 $$
 Then we distinguish four cases:
 \begin{enumerate}
\item $y_u=x_u$ and $y_v=x_v$.
In this case $K_e(\x,\y)=2\cdot \Phi_e(x_u,x_v)$ and thus
$K_e(\x,\y)=\Phi_e(\tx_u,\tx_v)+ \Phi_e(\ty_u,\ty_v)$.

\item $y_u \neq x_u$ and $y_v=x_v$.
In this case $K_e(\x,\y)=\Phi_e(x_u,x_v)+\Phi_e(y_u,x_v)$ and
thus $K_e(\x,\y)=\Phi_e(\tx_u,\tx_v)+ \Phi_e(\ty_u,\ty_v)$.

\item $y_u=x_u$ and $y_v \neq x_v$.
In this case $K_e(\x,\y)=\Phi_e(x_u,x_v)+\Phi_e(x_u,y_v)$ and
thus $K_e(\x,\y)=\Phi_e(\tx_u,\tx_v)+ \Phi_e(\ty_u,\ty_v)$.

\item $y_u \neq x_u$ and $y_v \neq x_v$.
In this case $K_e(\x,\y)=\Phi_e(y_u,x_v)+\Phi_e(x_u,y_v)$. Since $|\Phi_e(x_u,x_v) - \Phi_e(y_u,x_v)| \leq w_e$ and $|\Phi_e(y_u,y_v) - \Phi_e(x_u,y_v)| \leq w_e$, then
$$
 |K_e(\x,\y)-\Phi_e(\tx_u,\tx_v)-\Phi_e(\ty_u,\ty_v)| \leq 2w_e.
$$
\end{enumerate}
By summing the contribution of every edge we achieve \eqref{eq:K_generalG}.
\end{proof}
Note that the quantity $b(G)$ is related to the \emph{bipartiteness ratio} of $G$ which in turn is related
to the smallest eigenvalue of $G$ \cite{trevisanSTOC09}.

Finally, we next prove that for an $\alpha$-decomposable observable the extent at which the expectations
at stationarity for the one-logit and the all-logit differ depends only on $\alpha$ and $\beta$.
\begin{theorem}
If observable $O$ is decomposable then
$$e^{-2\alpha\beta} \cdot \langle O,\pi_1\rangle \leq \langle O,\pi_A\rangle \leq e^{2\alpha\beta} \cdot \langle O,\pi_1\rangle.$$
\end{theorem}
\begin{proof}
By mimicking the proof of Theorem~\ref{thm:bipdec},
we have $e^{-\alpha\beta} Z_1^2 \leq Z_A \leq e^{\alpha\beta} Z_1^2$ and
$$
e^{-\alpha\beta} \sum_\y\gamma_1(\mu_1(\x,\y))\cdot\gamma_1(\mu_2(\x,\y)) \leq \gamma_A(\x) \leq e^{\alpha\beta} \sum_\y\gamma_1(\mu_1(\x,\y))\cdot\gamma_1(\mu_2(\x,\y)).
$$
The theorem then follows from the same arguments given in the proof of Theorem~\ref{thm:suff}.
\end{proof}

\section{Mixing time}\label{sec::mixing}
The all-logit dynamics for a strategic game have the property that, for every pair of profiles $\x,\y$ and for every value of $\beta$, the transition probability from $\x$ to $\y$ is strictly positive. In order to give upper bounds on the mixing time, we will use the following simple well-known lemma (see e.g. Theorem 11.5 in~\cite{mu05}).
\begin{lemma}\label{lemma:tmixeasy}
Let $P$ be the transition matrix of an ergodic Markov chain with state space $\Omega$. For every $y \in \Omega$ let us name $\alpha_y = \min\{ P(x,y) \colon x \in \Omega \}$ and $\alpha = \sum_{y \in \Omega} \alpha_y$. Then the mixing time of $P$ is $\tm = \OO(1/\alpha)$.
\end{lemma}

We now give an upper bound holding for \emph{every game}.
Recall that for a strategic game $\mathcal{G}$,
in Section~\ref{sec::logit} we defined the cumulative utility function for the ordered pair of profiles $(\x,\y)$
as $U(\x,\y) = \sum_{i=1}^n u_i(\x_{-i},y_i)$. Let us name $\Delta U$ the size of the range of $U$,
$$
\Delta U = \max\{ U(\x,\y) \colon \x,\y \in S \} - \min\{ U(\x,\y) \colon \x,\y \in S \}.
$$
By using Lemma~\ref{lemma:tmixeasy} we can give a simple upper bound on the mixing time
of the all-logit dynamics for $\mathcal{G}$ as a function of $\beta$ and $\Delta U$.

\begin{theorem}[General upper bound]
\label{theorem:ubgeneral}
For any strategic game $\mathcal{G}$
the mixing time of the all-logit dynamics for $\mathcal{G}$ is $\OO\left( e^{\beta \Delta U} \right)$.
\end{theorem}
\begin{proof}
Let $P$ be the transition matrix of the all-logit dynamics for $\mathcal{G}$ and let $\x,\y \in S$ be two profiles. From~\eqref{eq:transition} we have that
$$
P(\x,\y) = \frac{e^{\beta U(\x,\y)}}{\sum_{\z \in S} e^{\beta U(\x,\z)}} = \frac{1}{\sum_{\z \in S} e^{\beta \left(U(\x,\z) - U(\x,\y) \right)}} \geqslant \frac{1}{|S| e^{\beta \Delta U}}.
$$
Hence for every $\y \in S$ it holds that
$$
\alpha_{\y} \geqslant \frac{e^{- \beta \Delta U}}{|S|}
$$
and $\alpha = \sum_{\y \in S} \alpha_{\y} \geqslant e^{-\beta \Delta U}$. The thesis then follows from Lemma~\ref{lemma:tmixeasy}.
\end{proof}
Next sections will give specific bounds for two specific classes of games (that contain the games analyzed in the Section~\ref{sec::examples}),
namely graphical coordination games and games with a dominant profile.
These results show that the the mixing time of the all-logit dynamics has the same twofold behavior that has been highlighted in the case
of the one-logit: for some games it depends exponentially on $\beta$, whereas for other games it can be upper-bounded by a function independent from $\beta$.

\subsection{Graphical coordination games}
A \emph{graphical coordination game} is a local interaction game in which on each edge is played the the coordination game described by~\eqref{eq:coorddef}.
It turns out that we can apply Theorem~\ref{theorem:ubgeneral} in order to give an upper bound to the mixing time of the all-logit for graphical coordination games.
\begin{theorem}
\label{thm:graphical}
 The mixing time of the all-logit for a graphical coordination game on a graph $G = (V,E)$ is
 $$
  \tm = \OO\left(e^{2\beta(\max\{a,b\}-\min\{c, d\})|E|}\right).
 $$
\end{theorem}
\begin{proof}
 Suppose that $a \geq b$. Then, consider the profile $\x_+$ in which each player plays the strategy $+1$.
 It is easy to see that $U(\x,\y) \leq U(\x_+,\x_+) = \sum_i a \cdot \deg(i)$, where $\deg(i)$ is the degree of $i$ in $G$.
 The case $a < b$ is equivalent except that we now consider the profile $\x_-$ in which each player plays the strategy $-1$.
 Similarly, suppose that $c \leq d$. Then $U(\x, \y) \geq U(\x_-, \x_+) = \sum_i c \cdot \deg(i)$.
 The case $d < c$ is equivalent except we invert the role of $\x_-$ and $\x_+$.
 Hence
 $$
 \Delta U = \sum_i \max\{a,b\} \cdot \deg(i) - \sum_i \min\{c,d\} \cdot \deg(i) = 2\beta(\max\{a,b\}-\min\{c, d\})|E|.
 $$
 The thesis then follows from Theorem~\ref{theorem:ubgeneral}.
\end{proof}
This bound shows that the mixing time of the all-logit for graphical coordination games exponentially depends on $\beta$,
as in the case of the one-logit dynamics.
However, the bounds given in the previous theorem can be very loose with respect to
the known results about the mixing time of the one-logit for graphical coordination games~\cite{afpppSPAA11}.
It would be interesting to understand at which extent the above bounds can be improved
(in Appendix~\ref{apx:cwmodel} we slightly improve these bounds for a very special graphical coordination game,
namely the Curie-Weiss model for ferromagnetism adopted in Statistical Physics) and, in particular,
if it is possible to show that the mixing time of the all-logit cannot be longer than the mixing time of the one-logit.

\subsection{Games with dominant strategies}
Theorems~\ref{thm:graphical} shows that for graphical coordination games the mixing time grows with $\beta$.
In this section we show that for games with a dominant profile, such as the prisoner's dilemma
analyzed in Section~\ref{sec::examples}, the time that the all-logit take for converging to the stationary distribution
is upper bounded by a function independent of $\beta$, as in the case of the one-logit dynamics~\cite{afpppSPAA11}.

Specifically, we say that strategy $s^\star \in S_i$ is a {\em dominant strategy} for player $i$
if for all $s' \in S_i $ and all strategy profiles $\x \in S$,
$$
 u_i(s^\star, \x_{-i}) \geq u_i(s', \x_{-i}).
$$
A {\em dominant profile} $\x^\star = (x^\star_1,\ldots,x^\star_n)$ is a profile in which $x^\star_i$ is a dominant strategy for player $i=1,\ldots,n$.
Then, we can derive the following upper bound on the mixing time of the all-logit dynamics for games with a dominant profile,
whose proof resembles the one used for proving a similar result for the one-logit given in~\cite{afpppSPAA11}.
\begin{theorem}
\label{theorem:dominant-strategies}
Let $\G$ be an $n$-player games with a dominant profile where each player has at most $m$ strategies.
The mixing time of the all-logit for $\G$ is
$$
 \tm=\OO\left(m^n\right).
$$
\end{theorem}
\begin{proof}
The proof uses the coupling technique (see, for example, Theorem~5.2 in \cite{lpwAMS08}).

Let $P$ be the transition matrix of the all-logit dynamics for $\G$.
For every pair of profiles $\x$ and $\y$, we consider a coupling $(X,Y)$ of the distributions $P(\x,\cdot)$ and $P(\y,\cdot)$
such that for every player $i$ the probability that both chains choose strategy $s$ for player $i$
is exactly $\min\{\sigma_i(s\mid\x),\sigma_i(s\mid\y)\}$.
Observe that, with such a coupling,  once the two chains coalesce, i.e. $X = Y$, they stay together.

We next observe that for all starting profiles $\x$ and $\y$, it holds that
$$
 \Prob{\x,\y}{X_1 = Y_1} \geq \Prob{\x,\y}{X_1 = \x^\star \mbox{ and } Y_1 = \x^\star} \geq \frac{1}{m^n}.
$$
Indeed both chains are in profile $\x^\star$ after one step if and only if every player chooses strategy $x_i^\star$ in both chains.
From the properties of the coupling, it follows that this event occurs with probability
$$
 \prod_i \min\{\sigma_i(x^\star_i \mid \x),\sigma_i(x^\star_i \mid \y)\} \geq \prod_i \frac{1}{|S_i|} \geq \frac{1}{m^n},
$$
where the first inequality follows from~\eqref{eq:logitrule} and the fact that $x^\star_i$ is a dominant strategy for $i$.

Therefore we have that the probability that the two chains have not yet coupled after $k$ time steps is
$$
\Prob{\x,\y}{X_{k} \neq Y_{k}} \leqslant \left( 1 - \frac{1}{m^n} \right)^k \leqslant
e^{- k/m^n},
$$
which is less than $1/4$ for $k = \OO(m^n)$. By applying the Coupling Theorem~\cite[Theorem~5.2]{lpwAMS08} we have that $\tm = \OO\left(m^n\right)$.
\end{proof}

\section{Conclusions and open problems}\label{sec::conclusions}
In this paper we considered the selection rule that assigns
positive probability only to the set of all players.
A natural extension of this selection rule
assigns a different probability to each subset of the players.
What is the impact of such a probabilistic selection rule
on reversibility and on observables?
Some interesting results along that direction have been
obtained in~\cite{afnGEB10,afn12}.
Notice that if we consider the selection rule that selects
player $i$ with probability $p_i>0$
(the one-logit set $p_i=1/n$ for all $i$) then the
stationary distribution is the same as the stationary distribution
of the one-logit. Therefore, all observables have the same expected value
and all potential games are reversible.

It is a classical result that the Gibbs distribution that is
the stationary distribution of the one-logit
(the micro-canonical ensemble, in  Statistical Mechanics parlance)
is the distribution that maximizes the entropy among all the distributions
with a fixed average potential. Can we say something similar for the
stationary distribution of the all-logit?
A promising direction along this line of research is suggested by
results in Section~\ref{sec:observables}:
at least in some cases the stationary distribution of the all-logit
can be seen as a composition of simpler distributions.

\newpage
\bibliographystyle{plain}
\bibliography{logit}

\newpage
\appendix
\section{Mixing time of the all-logit for the Curie-Weiss model}
\label{apx:cwmodel}
Here we prove upper and lower bounds on the mixing time of the all-logit dynamics for a special graphical coordination game, the {\sf CW}-game.
In such a game we set $a = b = +1$ and $c = d = -1$.
Thus, the utility of player $i \in [n]$ is the sum of the number of players playing the same strategy as $i$,
minus the number of players playing the opposite strategy;
that is, the utility of player $i \in [n]$ at profile $\x = (x_1, \dots, x_n) \in \{-1,+1\}^n$ is
$$
u_i(\x) = x_i \sum_{j\neq i} x_j.
$$
It is easy to see that the potential function for this game is
$$
\Phi(\x) = - \sum_{\{i,j\} \in \binom{[n]}{2}} x_i x_j.
$$
Due to the high level of symmetry of the game,
the potential of a profile $\x$ depends only on the \emph{number} of players playing $\pm 1$.
Indeed, we can rewrite the potential of $\x$ as
$$
\Phi(\x) = - \frac{\Diff^2(\x) - n}{2},
$$
where $\Diff$ is the observable described in Section~\ref{sec:dec}.

\paragraph{The upper bound.}
Observe that, for the Curie-Weiss model we have $\Delta U = 2n(n-1)$, hence by using Theorem~\ref{thm:graphical} we get directly that
\begin{equation}\label{eq:ubIsing1}
\tm = \OO\left( e^{2 \beta n (n-1)} \right).
\end{equation}
Hence it follows that mixing time is $\OO(1)$ for $\beta = \OO(1/n^2)$ and it is $\OO(\text{poly}(n))$ for $\beta = \OO(\log n / n^2)$.

In what follows we show that factor ``$2$'' at the exponent in~\eqref{eq:ubIsing1} can be removed and that a slightly better upper bound can be given for $\beta > \log n / n$.
\begin{lemma}\label{lemma:alphay}
For every $\x,\y \in \Omega$ it holds that
$$
P(\x,\y) \geqslant q^{(n+|\Diff(\y)|)/2} (1-q)^{(n-|\Diff(\y)|)/2}
$$
where
$$
q= \frac{1}{1 + e^{2 \beta (n-1)}}.
$$
\end{lemma}
\begin{proof}
Consider a profile $\y \in \{-1,+1\}^n$. Observe that the number of players playing $+1$ and $-1$ in $\y$ can be written as $\frac{n + \Diff(\y)}{2}$ and $\frac{n - \Diff(\y)}{2}$, respectively. If $\Diff(\y) > 0$, i.e. if the number of players playing $+1$ is larger than the number of players playing $-1$, then the profile that minimizes $P(\x,\y)$ is profile $\x_{-} = (-1, \dots, -1)$ where every player plays $-1$. If we name
$$
q = \frac{e^{- \beta (n-1)}}{e^{- \beta(n-1)} + e^{\beta(n-1)}} = \frac{1}{1+e^{2\beta(n-1)}}
$$
the probability that a player in $\x_{-}$ chooses strategy $+1$ for the next round, we have that
$$
P(\x_{-},\y) = q^{\frac{n + \Diff(\y)}{2}} (1-q)^{\frac{n - \Diff(\y)}{2}}.
$$
On the other hand, if $\Diff(\y) < 0$, then $P(\x,\y)$ is minimized when $\x = \x_{+} = (+1,\dots,+1)$ and, since $q$ is also the probability that a player in $\x_{+}$ chooses strategy $-1$ for the next round, we have that
$$
P(\x_{+},\y) = q^{\frac{n - \Diff(\y)}{2}} (1-q)^{\frac{n + \Diff(\y)}{2}}
$$
and the thesis follows.
\end{proof}

Now we can give an upper bound on the mixing time by using lemmata~\ref{lemma:tmixeasy} and~\ref{lemma:alphay}
\begin{theorem}[Upper bound]
The mixing time of the all-logit dynamics for the Curie-Weiss model is
$$
\tm = \OO\left( n e^{\beta n^2} \right).
$$
If $\beta \geqslant \log n / n$ the mixing time is
$$
\tm = \OO\left( \frac{n e^{\beta n^2}}{2^n} \right).
$$
\end{theorem}
\begin{proof}
From Lemma~\ref{lemma:alphay} it follows that for every $\y \in \{-1,+1 \}^n$ we have
$$
\alpha_{\y} = \min\{ P(\x,\y) \mid \x \in \{-1,+1\}^n \} \geqslant q^{(n+|\Diff(\y)|)/2} (1-q)^{(n-|\Diff(\y)|)/2}.
$$
Hence
\begin{equation}\label{eq:mix:alpha}
\alpha = \sum_{\y \in \{-1,+1\}^n} \alpha_{\y} \geqslant \sum_{\y \in \{-1,+1\}^n} q^{(n+|\Diff(\y)|)/2} (1-q)^{(n-|\Diff(\y)|)/2}.
\end{equation}
Now observe that there are $\binom{n}{\frac{n-k}{2}}$ profiles $\y$ such that $\Diff(\y) = k$, and since $q \leqslant 1/2$, the largest terms in~\eqref{eq:mix:alpha} are the ones such that $\Diff(\y)$ is as close to zero as possible. In order to give a lower bound to $\alpha$ we will thus consider only profiles $\y$ such that $\Diff(\y) = 0$, when $n$ is even, and profiles $\y$ such that $\Diff(\y) = \pm 1$, when $n$ is odd.

\noindent
\underline{Case $n$ even:} If we consider only profiles $\y$ such that $\Diff(\y) = 0$ in~\eqref{eq:mix:alpha} we have that
$$
\alpha \geqslant \binom{n}{n/2}[q(1-q)]^{n/2}.
$$
By using a standard lower bound for the binomial coefficient (see e.g. Lemma~9.2 in~\cite{mu05}) we have that
$$
\binom{n}{n/2} \geqslant \frac{2^n}{n+1}.
$$
As for $[q(1-q)]^{n/2}$ we have that
\begin{align}\label{eq:qq}
q(1-q) & = \frac{1}{1+e^{2\beta(n-1)}} \cdot \frac{1}{1+e^{-2\beta(n-1)}} \nonumber \\
& = \frac{1}{e^{2\beta(n-1)} + 2 + e^{-2\beta(n-1)}} \nonumber \\
& = \frac{1}{e^{2\beta(n-1)} \left(1 + 2 e^{-2\beta(n-1)} + e^{-4 \beta (n-1)}  \right)}
\end{align}
Now observe that for every $\beta \geqslant 0$ we can bound $1 + 2 e^{-2\beta(n-1)} + e^{-4 \beta (n-1)} \leqslant 4$. Thus we have that
\begin{equation}\label{eq:mix:boundqone}
[q(1-q)]^{n/2} \geqslant \frac{1}{2^n e^{\beta n(n-1)}}.
\end{equation}
Hence
$$
\alpha \geqslant \binom{n}{n/2}[q(1-q)]^{n/2} \geqslant \frac{1}{(n+1) e^{\beta n(n-1)}}.
$$
And by using Lemma~\ref{lemma:tmixeasy} we have
$$
\tm = \OO\left( n e^{\beta n(n-1)} \right).
$$
If $\beta$ is large enough, say $\beta \geqslant \log n / n$, in~\eqref{eq:qq} we can bound
$$
1 + 2 e^{-2\beta(n-1)} + e^{-4 \beta (n-1)} \leqslant 1 + \frac{1}{n}.
$$
Thus, in this case we have that
\begin{equation}\label{eq:mix:boundqtwo}
[q(1-q)]^{n/2} \geqslant  \frac{1}{e^{\beta n (n-1)} \left( 1 + 1/n \right)^{(n/2)}} \geqslant \frac{1}{e^{\beta n(n-1)} \cdot \sqrt{e}}.
\end{equation}
Hence $\alpha \geqslant \frac{2^n}{(n+1) e^{1/2 + \beta n(n-1)}}$ and
$$
\tm = \OO\left( \frac{n e^{\beta n(n-1)}}{2^n} \right).
$$

\noindent
\underline{Case $n$ odd:} If we consider only profiles $\y$ such that $\Diff(\y) = \pm 1$ in~\eqref{eq:mix:alpha} we get
$$
\alpha \geqslant 2 \binom{n}{\frac{n+1}{2}} q^{\frac{n+1}{2}} (1-q)^{\frac{n-1}{2}} = 2 \binom{n}{\frac{n+1}{2}} \left(q(1-q)\right)^{n/2} \sqrt{\frac{q}{1-q}}.
$$
Now observe that
$$
\sqrt{\frac{q}{1-q}} = e^{-\beta (n-1)} \qquad \mbox{ and } \qquad \binom{n}{\frac{n+1}{2}} \geqslant  \frac{1}{2} \cdot \frac{2^n}{n+1}.
$$
By using bounds~\eqref{eq:mix:boundqone} and \eqref{eq:mix:boundqtwo} for $[q(1-q)]^{n/2}$ we get $\tm = \OO\left( n e^{\beta (n^2 - 1)} \right)$ for every $\beta \geqslant 0$ and $\tm = \OO\left( \frac{n e^{\beta(n^2-1)}}{2^n} \right)$ for $\beta \geqslant \log n / n$.
\end{proof}

\paragraph{The lower bound.}
In order to give a lower bound on the mixing time, we first show that, for the Curie-Weiss model, $K(\x,\y)$ can be written as a function of $\Diff(\x)$, $\Diff(\y)$ and of the Hamming distance between the two profiles.
\begin{lemma}
Let $\x,\y \in \{-1,+1\}^n$ be two profiles with magnetization $\Diff(\x)$ and $\Diff(\y)$ respectively and let $h_{\x,\y}$ be their Hamming distance, i.e. the number of players where they differ. Then
$$
K(\x, \y) = n - \Diff(\x) \cdot \Diff(\y) - 2 h_{\x,\y}.
$$
\end{lemma}
\begin{proof}
As stated above, $\Phi(\x) = \frac{n - \Diff^2(\x)}{2}$. In order to evaluate $K(\x,\y) = \sum_{i=1}^n \Phi(\x_{-i}, y_i) - (n - 2) \Phi(\x)$ let us name $n_1,n_2$ and $n_3$ as follows
\begin{align*}
n_1 & = \# \{ i \in [n] \colon x_i = y_i \};\\
n_2 & = \# \{ i \in [n] \colon x_i = +1, y_i = -1 \};\\
n_3 & = \# \{ i \in [n] \colon x_i = -1, y_i = +1 \}.
\end{align*}
In other words, $n_1$ is the number of players playing the same strategy in profiles $\x$ and $\y$, $n_2$ is the number of players playing $+1$ in $\x$ and $-1$ in $\y$, and $n_3$ the number of players playing $-1$ in $\x$ and $+1$ in $\y$. It holds that
\begin{equation}
 \label{eq:evalpsi1}
\begin{aligned}
\sum_{i=1}^n \Phi(\x_{-i}, y_i) & = n_1 \frac{n- \Diff^2(\x)}{2} + n_2 \frac{n - (\Diff(\x) -2)^2}{2} + n_3 \frac{n - (\Diff(\x)+2)^2}{2}\\
& = \frac{1}{2} \left( (n_1+n_2+n_3) (n - \Diff^2(\x)) + 4(n_2-n_3) \Diff(\x) - 4 (n_2+n_3) \right).
\end{aligned}
\end{equation}
Now observe that $n_1+n_2+n_3 = n$, $2(n_2-n_3) = \Diff(\x) - \Diff(\y)$, and $(n_2+n_3) = h_{\x,\y}$. Hence from \eqref{eq:evalpsi1} we get
\begin{equation}
 \label{eq:evalpsi2}
 \begin{aligned}
\sum_{i=1}^n \Phi(\x_{-i}, y_i) & = \frac{1}{2} \left( n (n+\Diff^2(\x)) + 2(\Diff(\x)-\Diff(\y)) \Diff(\x) -4 h_{\x,\y}\right)\\
& = \frac{n^2}{2} - \frac{n-2}{2} \Diff^2(\x) - \Diff(\x)\Diff(\y) - 2 h_{\x,\y}.
\end{aligned}
\end{equation}
Thus
\[
K(\x,\y) = n - \Diff(\x) \cdot \Diff(\y) - 2 h_{\x,\y}. \qedhere
\]
\end{proof}
Since the Hamming distance between two profiles is at most $n$, from the above lemma we get the following observation.
\begin{obs}\label{obs:oppositemag}
Let $\x,\y$ be two profiles with $\Diff(\x) \cdot \Diff(\y) \leqslant 0$, then $K(\x,\y) \geq -n$.
\end{obs}

\noindent
Now we can give a lower bound on the mixing time by using the bottleneck-ratio technique.
\begin{theorem}[Lower bound]
The mixing time of the all-logit dynamics for the Curie-Weiss model is
$$
\tm = \Omega\left( \frac{e^{\beta n(n - 2)}}{4^n} \right).
$$
\end{theorem}
\begin{proof}
Let $S_- \subseteq \{-1,+1\}^n$ be the set of profiles $\x$ such that $\Diff(\x) < 0$, i.e.
$$
S_- = \{ \x \in \{-1,+1\}^n \colon \Diff(\x) < 0 \}
$$
and observe that $\pi(S_-) \leqslant 1/2$. From Observation~\ref{obs:oppositemag} we have that for every $\x \in S_-$ and $\y \in S_+ = \{-1,+1\}^n \setminus S_-$ it holds that
\begin{equation}\label{eq:mix:ubonq}
\pi(\x)P(\x,\y) = \frac{e^{-\beta K(\x,\y)}}{Z} \leqslant \frac{e^{\beta n}}{Z}.
\end{equation}
Moreover, if we name $\x_{-}$ the profile where everyone is playing $-1$ we have that
\begin{equation}\label{eq:mix:lbonpi}
\pi(S_-) \geqslant \pi(\x_{-}) \geqslant \frac{1}{Z} e^{-2\beta \Phi(\x_{-})} = \frac{1}{Z} e^{\beta n(n-1)}.
\end{equation}
Hence, by using bounds~\eqref{eq:mix:ubonq} and~\eqref{eq:mix:lbonpi}, and the fact that the size of $S_-$ is at most $2^{n-1}$, we can bound the bottleneck at $S_-$ with
$$
B(S_-) = \frac{Q(S_-,S_+)}{\pi(S_-)} = \frac{\sum_{\x \in S_-} \sum_{\y \in S_+} \pi(\x) P(\x,\y)}{\pi(S_-)} \leqslant \frac{2^{2n-2} e^{\beta n}}{e^{\beta n(n-1)}} = \frac{2^{2n-2}}{e^{\beta n (n-2)}}.
$$
By using the bottleneck-ratio theorem (see e.g. Theorem~7.3 in~\cite{lpwAMS08}) it follows that
\[
\tm = \Omega \left( \frac{e^{\beta n (n-2)}}{2^{2n}} \right).\qedhere
\]
\end{proof}

\paragraph{Remarks.}
In this section we proved upper and lower bounds on the mixing time of the all-logit dynamics for the Curie-Weiss model. In particular, the upper bound shows that for $\beta = \OO(1/n^2)$ the mixing time is constant and for $\beta = \OO( \log n / n^2)$ it is at most polynomial. The lower bound shows that, for every constant $\varepsilon > 0$, if $\beta > (1+\varepsilon) (\log 4) /n$ the mixing time is exponential. When $\beta$ is between $\Theta(\log n /n^2)$ and $\Theta(1/n)$ we still cannot say if mixing is polynomial or exponential.

\end{document}